\documentclass[11pt]{article}
\usepackage[english]{babel}
\usepackage[letterpaper,top=0.95in,bottom=1.05in,left=1in,right=1in,marginparwidth=2cm]{geometry}
\usepackage[T1]{fontenc}   
\usepackage{amsmath, amssymb, amscd, amsthm, amsfonts}
\usepackage{thm-restate}
\usepackage{graphicx}
\usepackage[colorlinks=true, allcolors=blue]{hyperref}
\usepackage{tikz}
\usepackage{cellspace}
\usepackage{cleveref}
\usepackage{xspace}
\usepackage{csquotes}
\usepackage{enumitem}

\usepackage{multirow}

\usepackage{booktabs}
\usepackage{ragged2e}
\usepackage{framed}
\usepackage[framemethod=tikz]{mdframed}

\usepackage[style=trad-alpha,natbib=true,maxcitenames=3]{biblatex}
\title{Overlay Network Construction: \\ Improved Overall and Node-Wise Message Complexity}

\author{Yi-Jun Chang\footnote{National University of Singapore. Email: cyijun@nus.edu.sg} \and Yanyu Chen\footnote{National University of Singapore. Email: yanyu.chen@u.nus.edu} \and Gopinath Mishra\footnote{National University of Singapore. Email: gopinath@nus.edu.sg}}

\date{}

\usepackage[colorinlistoftodos,prependcaption,textsize=tiny, disable]{todonotes}

\newcommand{\yanyu}[1]{\todo[backgroundcolor=red!25]{Yanyu: #1}}

\newtheorem{theorem}{Theorem}[section]

\newtheorem{lemma}[theorem]{Lemma}

\newtheorem{observation}[theorem]{Observation}
\newtheorem{corollary}[theorem]{Corollary}

\newcommand{\EE}{\mathbb{E}}
\newcommand{\floor}[1]{\left \lfloor #1 \right \rfloor}

\newcommand{\polylog}{\operatorname{polylog}}

\newcommand{\vol}{\operatorname{vol}}
\newcommand{\id}{\operatorname{id}}

\newcommand{\PCONGEST}{\textnormal{\textsf{P2P}-\textsf{CONGEST}}}
\newcommand{\GOSSIP}{\textnormal{\textsf{GOSSIP}}}
\newcommand{\GOSSIPr}{\textnormal{\textsf{GOSSIP-reply}}}

\newcommand{\GOSSIPrOne}{$\GOSSIPr\left(\log n\right)$}

\newcommand{\NCC}{\textnormal{\textsf{NCC}}}
\newcommand{\HYBRID}{\textnormal{\textsf{HYBRID}}}
\newcommand{\CONGEST}{\textnormal{\textsf{CONGEST}}}
\newcommand{\LOCAL}{\textnormal{\textsf{LOCAL}}}
\newcommand{\WFT}{\textnormal{WFT}}
\newcommand{\CExp}{\textnormal{\textsc{CreateExpander}}}
\newcommand{\ExpDR}{\textnormal{\textsc{ExpanderDegreeReduction}}}
\newcommand{\hybridWFT}{\textnormal{\textsc{HybridWFT}}}
\newcommand{\MergeStar}{\textnormal{\textsc{MergeStar}}}

\newcommand{\Red}{\textnormal{\textsf{Red}}}
\newcommand{\Blue}{\textnormal{\textsf{Blue}}}

\newcommand{\FindOut}{\textsc{FindOutgoing}}
\newcommand{\HPTestOut}{\textsc{HPTestOut}}
\newcommand{\RCtT}{\textsf{RC2T}}

\makeatletter
\newcommand*{\whp}{%
    \@ifnextchar{.}%
        {w.h.p}%
        {w.h.p.\@\xspace}%
}
\makeatother

\newcommand{\wft}{{well-formed tree}}
\newcommand{\satt}{{satisfactory tree}}
\newcommand{\Out}{\textnormal{Out}}
\newcommand{\CC}{\textnormal{\textsf{CC}}}
\newcommand{\cc}{\textnormal{\textsf{cc}}}
\newcommand{\poly}{\textnormal{poly}}
\newcommand{\ie}{i.e., }

\begin{document}
\pagenumbering{gobble}
\maketitle

\begin{abstract}
We consider the problem of constructing distributed overlay networks, where nodes in a reconfigurable system can create or sever connections with nodes whose identifiers they know. Initially, each node knows only its own and its neighbors' identifiers, forming a local channel, while the evolving structure is termed the global channel. The goal is to reconfigure any connected graph into a desired topology, such as a bounded-degree expander graph or a well-formed tree (WFT) with a constant maximum degree and logarithmic diameter, minimizing the total number of rounds and message complexity. This problem mirrors real-world peer-to-peer network construction, where creating robust and efficient systems is desired.

We study the overlay reconstruction problem in a network of $n$ nodes in two models:  \textsf{GOSSIP-reply}{} 
and \textsf{HYBRID}{}. In the \textsf{GOSSIP-reply}{} model, each node can send a message and receive a corresponding reply message in one round. In the \textsf{HYBRID}{} model, a node can send $O(1)$ messages to each neighbor in the local channel and a total of $O(\log n)$ messages in the global channel. 

In both models, we propose protocols for WFT construction with $O\left(n \log n\right)$ message complexities using messages of $O(\log n)$ bits. In the \textsf{GOSSIP-reply}{} model, our protocol takes $O(\log n)$ rounds while in the \textsf{HYBRID} model, our protocol takes $O(\log^2 n)$ rounds. 
Both protocols use $O\left(n \log^2 n\right)$ bits of communication. 

We obtain improved bounds over prior work:

\begin{description}
    \item[\textsf{GOSSIP-reply}{}:] A recent result by Dufoulon et al.~(ITCS 2024)  achieved $O\left(\log^5 n\right)$ round complexity and $O\left(n \log^5 n\right)$ message complexity using messages of at least $\Omega\left(\log^2 n\right)$ bits in \textsf{GOSSIP-reply}{}. With messages of size $O\left(\log n\right)$, our protocol achieves an optimal round complexity of $O(\log n)$ and an improved message complexity of $O(n \log n)$.
    \item[\textsf{HYBRID}{}:] Götte et al. (Distributed Computing 2023) showed an optimal $O(\log n)$-round algorithm with $O(\log^2 n)$ global messages per round  
    which incurs a message complexity of $\Omega(m)$, where $m$ is the number of edges in the initial topology.
    At the cost of increasing the round complexity to $O(\log^2 n)$ while using only $O(\log n)$ messages globally, our protocol achieves a message complexity that is independent of $m$. Our approach ensures that the total number of messages for node $v$, with degree $\deg(v)$ in the initial topology, is bounded by $O\left(\deg(v) + \log n\right)$, while the algorithm of Götte et al.~requires $O\left(\deg(v) + \frac{\log^4 n}{\log \log n}\right)$ messages per node.
\end{description}

\end{abstract}
\newpage
\tableofcontents

\clearpage
\pagenumbering{arabic}
\section{Introduction}




Many of today’s large-scale distributed systems on the Internet, such as peer-to-peer (P2P) and overlay networks, prioritize forming logical networks over relying (only) on the physical infrastructure of the underlying network. In these systems, direct connections between nodes can be virtual, using the physical connections of the Internet, and nodes are considered connected if they know each other’s IP addresses, allowing them to communicate and establish links. Examples of such systems include cryptocurrencies, the Internet of Things, the Tor network, and overlay networks like Chord \cite{stoica2001chord}, Pastry \cite{rowstron2001pastry}, and skip graphs \cite{aspnes2003skip}. These networks have the flexibility to reconfigure themselves by choosing which connections to establish or drop. This work focuses on the challenge of efficiently constructing a desired overlay network from any starting configuration of $n$ nodes, recognizing that the problem has a lower bound of $O(\log n)$, since even in an optimal scenario, it takes at least $O(\log n)$ rounds for two endpoints to connect if the nodes initially form a line.

In this paper, we address the well-explored challenge of efficiently constructing overlay topologies in a distributed manner within reconfigurable networks. This task is crucial in modern P2P networks, where topological properties are vital in ensuring optimal performance. Over the past two decades, numerous theoretical studies \cite{PRU01, LS03, GMS06, CDG07, JP12,augustine2015enabling}, have focused on developing P2P networks that exhibit desirable characteristics such as high conductance, low diameter, and resilience to substantial adversarial deletions. The common approach in these studies is to build a bounded-degree random graph topology in a distributed way, which ensures these properties. This approach leverages the fact that random graphs are likely to be expanders, possessing all the desired attributes \cite{MU05, HLW06}. Random graphs have been extensively used to model P2P networks \cite{PRU01, LS03, MU05, MS06, CDG07, APRU12, APR13, AMMPRU13, APR15}, and the random connectivity topology has been widely adopted in many contemporary P2P systems, including those underpinning blockchains and cryptocurrencies like Bitcoin \cite{MDV20}.

Several works have focused on overlay construction that transforms an arbitrary connected graph into a desired topology \cite{angluin2005fast, gmyr2017distributed, gotte2019faster, gotte2023time, dufoulon2024time}. While minimizing the number of rounds is the primary objective, reducing the total number of messages exchanged (message complexity) is also crucial. The message complexity in \cite{dufoulon2024time} is $O(n \log^5 n)$, whereas in other works it is $\Omega(m)$, where $n$ is the number of nodes and $m$ is the number of edges in the initial topology. In this work, we propose protocols in two models, \GOSSIPr{} and \HYBRID{} (defined formally later), that are both round-efficient and communication-efficient. Additionally, our protocol in the \HYBRID{} model optimizes the node-wise message complexity (i.e., the number of messages each node sends and receives based on its degree) compared to the previous work \cite{gotte2023time}.


Before discussing our results and comparing them with previous works, we formally introduce the models in the next section.
\subsection{Models}


We consider synchronous models on a fixed set of nodes $V$, where each node $v\in V$ has a unique identifier $\id(v)$ of length $O(\log n)$, with $n = |V|$. The \emph{local/input} network is represented by a graph $G = (V, E)$. Without loss of generality, we assume that $G$ is connected. Computation proceeds in synchronous rounds, during which the \emph{global/overlay} network evolves. In each round, nodes can send and receive messages and perform local computations. Unless otherwise specified, messages are $O(\log n)$ bits.

The network is \textit{reconfigurable} in the sense that if $u$ knows the identifier of $v$, then $u$ can send a reconfiguration message to $v$ to establish or drop the communication link. 
Lastly, we also allow implicit edge deletion since it can be easily implemented by only keeping edges established after round $r$. 

In general, we assume that the communication links are reliable, and no messages are dropped when the message capacity of the link is not exceeded. 
We also assume that there is sufficient memory on each computing node for our algorithm to run correctly and is capable of processing all incoming messages at the start of a round within that same round.



In this paper, as already mentioned, we consider two synchronous models: \GOSSIPr{} and \HYBRID{}, formally defined as follows.
%


\paragraph{\GOSSIPr{} model:} 
One of the earliest works in overlay construction by \citeauthor{angluin2005fast} \cite{angluin2005fast} considered a model---now known as the \GOSSIP{} model---where each node is allowed to send only one message per round. Recently, a reply version of this model, the $\GOSSIPr \left(b\right)$ model, was introduced by \cite{dufoulon2024time}, where they developed the first communication-efficient protocol.\footnote{In \cite{dufoulon2024time}, it is referred to simply as the \GOSSIP{}-based model or the \textsf{P2P-GOSSIP} model.} In this model, each node $v$ can perform the following actions in one round:
\begin{enumerate}[topsep=1.5ex, itemsep=0.1em]
    \item Send a message of $O(b)$ bits to a neighbor, where any node whose identifier is known to $v$ is considered a neighbor of $v$. We call this message the \textit{contacting message}.
    \item Receive all messages sent to $v$. Do some local computation.\footnote{Although local computation in Step~2 is not explicitly included in the original definition of \cite{dufoulon2024time}, preparing the outgoing messages appears to require some. Nevertheless, two rounds of the model without local computation in Step~2 suffice to simulate one round of the model with it.}
    \item Send an $O(b)$-bit reply to each of the contacting messages.
    \item Receive an $O(b)$-bit reply. Do some local computation. 
\end{enumerate}
Observe that in this model, there will be at most $2n$ messages sent in each round, $n$ contacting messages and $n$ replying messages. Hence, any algorithm with $O(T)$ round complexity has $O(nT)$ message complexity.

\paragraph{\HYBRID{} model:} The hybrid model was proposed in \cite{augustine2020hybrid} to study shortest path problems and later considered by \cite{gotte2023time} in the context of overlay construction problem.  In this model, the communication is done over both local and global channels. 
The $\HYBRID \left(\alpha, \beta, \gamma\right)$ model is defined by three parameters $\alpha$, $\beta$, and $\gamma$: 
\begin{itemize}[topsep=1.5ex, itemsep=0.1em]
    \item Each message size is $O(\alpha)$ bits.
    \item  Each node can send and receive $O(\beta)$ messages per round to each local neighbor, i.e., the local capacity is $O(\beta)$.
    \item  Each node can send and receive $O(\gamma)$ messages per round to any node whose identifier it knows, i.e., the global capacity is $O(\gamma)$.\footnote{All $O(\gamma)$ messages can be directed to a single global neighbor or spread across $O(\gamma)$ global neighbors.}
\end{itemize}

A subtle aspect of the model arises when a node is sent more messages in a round than its capacity permits. A standard assumption is that the node receives an arbitrary subset of these messages while the rest are dropped. All our algorithms guarantee that, with high probability, this situation never occurs, so the exact handling of such cases is irrelevant to our results.

In this work, we study the $\HYBRID(\log n, 1, \log n)$ model, where each message has size $O(\log n)$ bits, the local capacity is $O(1)$, and the global capacity is $O(\log n)$. Importantly, a node may send a number of local messages proportional to its degree, whereas in the global network it is limited to $O(\log n)$ messages. The rationale for assigning different capacities lies in the cost assumption: local communications, which take place on the given topology, are cheaper, while global communications, which require establishing new links, are more costly.

In addition to the \GOSSIPr{} and \HYBRID{} models, the \PCONGEST{} model has also been widely studied \cite{gmyr2017distributed, gotte2019faster, gotte2023time}. In \PCONGEST{}, each node can send and receive up to $O(\Delta \log n)$ messages of $O(\log n)$ bits per round, where $\Delta$ is the maximum degree of any node in the initial topology. It is important to note that $\mbox{\HYBRID{}} \left(\log n, 1, \log n\right)$ is \emph{weaker} than \PCONGEST{} in the sense that any protocol in $\mbox{\HYBRID{}} \left(\log n, 1, \log n\right)$ can be simulated in the \PCONGEST{} model without asymptotically increasing the round or message complexity asymptotically. Therefore, all results obtained in $\mbox{\HYBRID{}} \left(\log n, 1, \log n \right)$ also apply to \PCONGEST{}. 





{




\subsection{Our contribution and comparison with prior work}
As already mentioned, we focus on designing algorithms in both the   \GOSSIPr{} and \HYBRID{}   models that are efficient in terms of rounds and communication. Additionally, we demonstrate that our protocol for the \HYBRID{} model also achieves improved node-wise message complexity. We also discuss implications for the \PCONGEST{} model. See \Cref{app:model} for a discussion on the tradeoffs between different complexity measures and the motivation for studying node-wise message complexity.

Unless otherwise specified explicitly, all of our results including prior works are randomized and succeed \emph{with high probability} (\whp), i.e., with probability at least $1-1/\poly(n)$.




\subsubsection{Our result in the \GOSSIPr{} model}

Our result in the \GOSSIPr{} model is summarized in the following theorem. It shows that starting from any arbitrary topology, we can transform it into a star overlay. We then demonstrate how a star overlay can be converted into a desired topology by leveraging the properties of the \GOSSIPr{} model.

\begin{restatable}{thm}{star}
    \label{thm:star}
    There is a protocol in the $\mbox{\GOSSIPr} \left(b\right)$ model that can construct a star overlay in $O\left(\log n \cdot\max\left(\frac{\log n}{b}, 1\right)\right)$ rounds with $O\left(n \log n \cdot\max\left(\frac{\log n}{b}, 1\right)\right)$ messages \whp.
\end{restatable}

Note that building a star topology is similar to doing leader election in the \textit{reconfigurable} network. Observe that when the star topology is constructed, we can treat the distinguished center node in the star as the leader and perform many tasks on the leader node locally. More specifically, in $\GOSSIPr(b)$, we can reconfigure the network from the star topology to any topology whose maximum degree is $\Delta = O\left(\frac{b}{\log n}\right)$ in $O(1)$ round.

\begin{observation}
    \label{lm:central}
    If the initial topology $G$ is a star, then there is an $O(\Delta(H))$-round protocol in the $\mbox{\GOSSIPr} \left(b\right)$ model to construct an overlay network with a desired topology $H$ whose maximum degree is $\Delta(H) = O\left(\frac{b}{\log n}\right)$.
\end{observation}

\begin{proof}
    Firstly, every node except the distinguished center node sends a request message to the center node. Then the center node will compute an assignment of the nodes in the desired topology locally and reply to every node $v$ with their neighborhood $N_{H}(v)$. Each node then takes $O\left(\Delta(H)\right)$ rounds to establish connections with the new neighbors formally. Note that the center node needs to send $O\left(\Delta(H) \log n\right)$ bits to every leaf node. Since $\Delta(H) = O\left(\frac{b}{\log n}\right)$, the information that the center node needs to send is $O\left(\Delta(H) \log n\right) = O(b)$ bits, which can be sent in one message.
\end{proof}

We obtain the following corollary by applying \Cref{lm:central} and \Cref{thm:star} with $b=O(\log n)$.
\begin{corollary}
    \label{cor:gossip}
    There is a protocol in the \GOSSIPrOne{}  model that can construct any constant degree overlay network in $O(\log n)$ rounds with $O(n \log n)$ messages \whp.
  
\end{corollary}



\paragraph{Comparison with \Cite{dufoulon2024time}:} The algorithm by \citet{dufoulon2024time}  in the $\mbox{\GOSSIPr}\left(b\right)$ model, with $b=\Omega(\log^2 n)$, converts any arbitrary topology into a constant-degree expander in $O(\log^5 n)$ rounds, with a message complexity of $O(n \log^5 n)$. Thus, \Cref{cor:gossip}
 provides a strict improvement over \cite{dufoulon2024time} in both round and message complexity. Additionally, our algorithm can produce any constant-degree overlay, whereas the algorithm of \cite{dufoulon2024time} could only construct a constant-degree expander. The comparison of our result in the $\mbox{\GOSSIPr}\left(b\right)$ model with that of \cite{dufoulon2024time} is also presented in \Cref{tab:gossip}.


\begin{table}[htb]
    \centering
    \small
    \renewcommand{\arraystretch}{1.2}
    \setlength{\tabcolsep}{6pt} 
     \caption{Improvements in the $\mbox{\GOSSIPr} \left(b\right)$ model.}
     \vspace{5pt}
    \label{tab:gossip}
    \begin{tabular}{c|c|c|c|c}
        \hline \hline
        
        Reference &  $b$ & Rounds & Total message complexity & Target topology \\
        
        \hline \hline 
      
        \cite{dufoulon2024time} & $\Omega(\log^2 n)$ & $O(\log^5 n)$ & $O(n\log^5 n)$ & $O(1)$-degree expander \\
         
        \hline
         
        \Cref{cor:gossip} & $O(\log n)$ & $O(\log n)$ & $O(n\log n)$ & Any $O(1)$-degree graph \\
        
       
      
        \hline \hline 
    \end{tabular}
   
\end{table}


\subsubsection{Our result in the \HYBRID{} model}

In the $\mbox{\HYBRID}\left(\log n, 1, \log n\right)$ model, our main result is summarized in the following theorem: we show that starting from any arbitrary initial topology, it is possible to transform the network into a   \emph{well-formed tree} (WFT) efficiently. A \wft{} with $n$ nodes is defined as one that has a constant maximum degree and a depth of $O(\log n)$. Furthermore, we discuss how a \wft{} can be efficiently converted into a constant-degree expander in $\mbox{\HYBRID{}}(\log n, 1, \log n)$ model using the results from prior work in \cite{dufoulon2024time, gotte2023time}.


\begin{restatable}{thm}{hybridwft}
    \label{thm:hybridwft}
    There is a protocol in the $\mbox{\HYBRID}\left(\log n, 1, \log n\right)$ model that can construct a well-formed tree overlay from any input graph $G$ in $O\left(\log^2 n\right)$ rounds with $O\left(n \log n\right)$ messages \whp. Moreover, each node $v$ sends and receives at most $O(\deg_G(v) + \log n)$ messages throughout the protocol.
\end{restatable}


We remark that, although the sum of the node-wise bounds appears to imply an $O(m)$ message complexity, in our algorithm only a small subset of nodes may incur as many as $\Omega(\deg_G(v))$ messages, and the message complexity remains bounded by $O(n \log n)$. Due to the use of  randomness, it is not possible to determine in advance which nodes incur these higher costs.

Our algorithm follows a Boruvka-style cluster-merging process while maintaining the invariant that each cluster induces a well-formed tree. Outgoing edges are identified using sketching techniques. To achieve the node-wise message bound of $O(\deg_G(v) + \log n)$, we address the high communication load on star centers during cluster merges by introducing a matching-based method that pairs clusters for merging, even without a direct connecting edge, while ensuring that the total number of clusters reduces by a constant factor in each round of the merging process.   

To further optimize the message complexity, we employ a randomized procedure for constructing an $O(\log n)$-degree, $O(\log n)$-depth tree from a cycle. This improves upon the deterministic pointer-jumping process of prior work, yielding an $O(\log n)$-factor reduction in message cost.


Through minor modifications of the works in \cite{dufoulon2024time, gotte2023time}, we restate the following lemma, which enables the efficient transformation of a constant degree overlay network (e.g., a well-formed tree) into a constant degree expander network with constant conductance \whp in the $\HYBRID(\log n, 1, \log n)$ model. 

\begin{restatable}{lm}{wfttoexp}\textnormal{\cite{dufoulon2024time,gotte2023time}}
    \label{lm:wft-expander}
    Consider the $\mbox{\HYBRID}(\log n, 1, \log n)$ model. For any constant $\Phi \in (0,1/10]$, there is a protocol that takes $O\left(\log^2 n\right)$ rounds and $O\left(n \log^2 n\right)$ messages to convert an $O(1)$-degree overlay network into an $O(1)$-degree expander network with conductance at least $\Phi$, \whp. Moreover, each node $v$ sends and receives at most $O\left(\frac{\log^3 n}{\log \log n}\right)$ messages \whp.
\end{restatable}

The proof of \Cref{lm:wft-expander} is provided in \Cref{sec:expander} for completeness.
By applying \Cref{lm:wft-expander} after \Cref{thm:hybridwft}, we can construct an expander overlay network in the $\mbox{\HYBRID}(\log n, 1, \log n)$ model in $O(\log^2 n)$ rounds with $O(n \log^2 n)$ messages.

\begin{corollary}
    \label{cor:hybridwf}
    There is a protocol in the $\mbox{\HYBRID}(\log n, 1, \log n)$ model with the following guarantees:
\begin{itemize}
    \item It constructs a constant-degree expander graph from any input graph $G$ in $O\left(\log^2 n\right)$ rounds using $O\left(n \log^2 n\right)$ messages \whp.
    \item Each node $v$ sends and receives at most $O\left(\deg_G(v) + \frac{\log^3 n}{\log \log n}\right)$ messages.
\end{itemize}

\end{corollary}


\paragraph{Comparison with \cite{gotte2023time}:} \citet{gotte2023time} investigated the problem of overlay reconstruction in the $\mbox{\HYBRID}(\log n, 1, \log^2 n)$ model, aiming to convert an arbitrary initial topology into a well-formed tree or a constant-degree expander. Their algorithm achieved an optimal round complexity of $O(\log n)$ rounds with a message complexity of $\Omega(m + n \log^3 n)$, and the maximum number of messages sent or received by a node of degree $\deg_G(v)$ is $O\left(\deg_G(v) + \frac{\log^4 n}{\log \log n}\right)$, see \Cref{sec:nodewiseprior}. 
In comparison, although our result in \Cref{thm:hybridwft} and \Cref{cor:hybridwf}
} require $O(\log^2 n)$ rounds, it operates in the weaker $\mbox{\HYBRID}(\log n, 1, \log n)$ model. Crucially, the message complexity of our algorithm does not depend on $m$, and it achieves better node-wise message complexity compared to \cite{gotte2023time}. It is important to note, however, that our approach does not lead to an $O(\log n)$-round algorithm even in the $\mbox{\HYBRID}(\log n, 1, \log^2 n)$ model. 
The comparison of our result in the \HYBRID{} model with that of \cite{gotte2023time} is presented in \Cref{tab:hybrid}.
An open question remains: is it possible to achieve the optimal $O(\log n)$ rounds in the $\HYBRID(\log n, 1, \log n)$ model (or even in the $\HYBRID(\log n, 1, \log^2 n)$ model) with a message complexity of $O(n \cdot \text{poly}(\log n))$?


\begin{table}[htb]
    \centering
    \scriptsize
    \renewcommand{\arraystretch}{1.6}
    \setlength{\tabcolsep}{6pt} 
    
     \caption{Improvements in the $\mbox{\HYBRID} \left(\log n, 1, \gamma\right)$ model.}   \label{tab:hybrid}
     \vspace{5pt}
 
    \begin{tabular}{c|c|c|c|c|c}
        \hline \hline
        
   \multirow{2}{*}{Reference} & \multirow{2}{*}{$\gamma$} & \multirow{2}{*}{Rounds} & \multicolumn{2}{c|}{Message complexity} & \multirow{2}{*}{Target topology}\\
        \cline{4-5}
        & & & Total & Node-wise &  \\

        \hline \hline 
      
        \cite{gotte2023time} & $O(\log^2 n)$ & $O(\log n)$ & $\Omega(m + n \log^3 n)$ & $O\left(\deg_G(v)+\frac{\log^4 n}{\log\log n} \right)$ & \WFT/$O(1)$-degree expander \\
         
        \hline
         
        \Cref{thm:hybridwft}  & $O(\log n)$ & $O(\log^2 n)$ & $O(n\log n)$ & $O\left(\deg_G(v)+\log n \right)$ & \WFT \\

         \hline
         
        \Cref{cor:hybridwf}  & $O(\log n)$ & $O(\log^2 n)$ & $O(n\log^2 n)$ & $O\left(\deg_G(v)+\frac{\log^3 n}{\log \log n} \right)$ & $O(1)$-degree expander  \\
      
        \hline \hline 
    \end{tabular}
   
\end{table}

    
  

\subsubsection{Implication in \PCONGEST{} model}
\citet{gotte2023time} considered two models: \PCONGEST{} and $\mbox{\HYBRID}\left(\log n, 1, \log^2 n\right)$, which are not directly comparable. In particular, the result of \citet{gotte2023time} in $\mbox{\HYBRID{}}\left(\log n, 1, \log^2 n\right)$ does not translate directly to \PCONGEST{}. However, as already mentioned before, any protocol in  $\mbox{\HYBRID{}} \left(\log n, 1, \log n \right)$ can be simulated in the \PCONGEST{} model without asymptotically increasing the round or message complexity. Therefore, from \Cref{thm:hybridwft} and \Cref{cor:hybridwf}, we obtain the following corollary.


\begin{corollary}
    \label{cor:p2p}
     There is a protocol in the \PCONGEST{} model that can construct a \wft{} or a constant-degree expander graph from any input graph $G$ in $O\left(\log^2 n\right)$ rounds. The algorithm uses $O\left(n \log n\right)$ messages or $O\left(n \log^2 n\right)$ \whp for \wft{} or a constant-degree expander graph, respectively. Moreover, each node $v$ sends and receives at most $O(\deg_G(v) + \log n)$  or $O\left(\deg_G(v) + \frac{\log^3 n}{\log \log n}\right)$ messages depending on whether the target topology is a \wft{} or a constant-degree expander graph, respectively.
\end{corollary}


\paragraph{Comparison with \cite{gotte2023time}:}
\citet{gotte2023time} studied the overlay reconstruction problem in the \PCONGEST{} model, where the goal was to transform an arbitrary initial topology into a well-formed tree. Their solution achieved an optimal round complexity of $O(\log n)$ with a message complexity of $\Omega(m \log n + n \log^2 n)$, and the maximum number of messages sent or received by a node of degree $\deg_G(v)$ is $O\left(\deg_G(v)\cdot \frac{\log^3 n}{\log \log n}\right)$, see \Cref{sec:nodewiseprior}. 
In contrast, our algorithm, as stated in \Cref{cor:p2p}, requires $O(\log^2 n)$ rounds but notably achieves message complexity independent of $m$ and improved node-wise message complexity.  The comparison of our implication in the \PCONGEST{} model with that of the result of \cite{gotte2023time} is also presented in \Cref{tab:p2p}. A key remaining open question is whether it is possible to attain the optimal $O(\log n)$ rounds in the \PCONGEST{} model with a message complexity of $O(n \cdot \mbox{poly}(\log n))$.

\begin{table}[htb]
    \centering

    \renewcommand{\arraystretch}{1.6}
    \setlength{\tabcolsep}{6pt} 
    
     \caption{Improvements in \PCONGEST{} model.}   \label{tab:p2p}
     \vspace{5pt}
 
    \begin{tabular}{c|c|c|c|c}
        \hline \hline

        \multirow{2}{*}{Reference} & \multirow{2}{*}{Rounds} & \multicolumn{2}{c|}{Message complexity} & \multirow{2}{*}{Target topology}\\
        \cline{3-4}
        &  & Total & Node-wise &  \\

        \hline \hline 
      
        \cite{gotte2023time}  & $O(\log n)$ & $\Theta(m\log^2 n)$ & $O\left(\deg_G(v)\cdot\frac{\log^3 n}{\log \log n}\right)$ & \WFT/ $O(1)$-degree expander\\
         
        \hline
         
        \Cref{cor:p2p}   & $O(\log^2 n)$ & $O(n\log n)$ & $O\left(\deg_G(v)+\log n\right)$ & \WFT\\

         \hline
         
        \Cref{cor:p2p}   & $O(\log^2 n)$ & $O(n\log^2 n)$ & $O\left(\deg_G(v)+\frac{\log^3 n}{\log \log n}\right)$ & $O(1)$-degree expander\\
      
        \hline \hline 
    \end{tabular}
   
\end{table}
\subsection{Number of bits communicated}

For all of \Cref{thm:star}, \Cref{cor:gossip}, \Cref{thm:hybridwft}, and \Cref{cor:hybridwf}, the total number of bits communicated among the nodes in all protocols is $O(n \log^2 n)$, due to the use of the hashing techniques from \citet{king2015construction}. This bound on the communication complexity breaks the $\Omega(n \log^3 n)$  barrier of the {linear sketch} of \citet{ahn2012analyzing} which was used in the previous work \cite{dufoulon2024time}. 

It has been shown that  $\Omega(n \log^3 n)$ bits of communication are indeed necessary for several applications of the {linear sketch} of Ahn, Guha, and McGregor~\cite{ahn2012analyzing}, such as distributed and sketching spanning forest~\cite{nelson2019optimal} and connectivity~\cite{yu2021tight}. More concretely, in the distributed sketching model, the goal of the connectivity problem is to determine whether an $n$-node graph $G$ is connected, with each of the $n$ players having access to the neighborhood of a single vertex. Each player sends a message to a central referee, who then decides whether $G$ is connected. \citet{yu2021tight} established that for the referee to decide correctly with probability ${1}/{4}$, the total communication must be at least $\Omega(n \log^3 n)$ bits.

The ability to break this barrier stems from being able to communicate in both ways in multiple rounds as opposed to the one-round one-way setting described above. The optimality of the hashing technique from \citet{king2015construction} is much less well-known, and it remains open whether their communication complexity bound can be further improved. Any such improvement will likely lead to improvements in the $O(n \log^2 n)$ communication complexity bound in this paper as well as many other applications of the hashing technique.

\subsection{Related work}
Various studies have explored methods for transforming arbitrary connected graphs into specific target topologies, such as expanders and \wft{} \cite{angluin2005fast, gmyr2017distributed, gotte2019faster, gotte2023time, dufoulon2024time}. Angluin et al.~\cite{angluin2005fast} were among the first to address this problem, demonstrating that any connected graph $G$ with $n$ nodes and $m$ edges can be converted into a binary search tree with depth $O(\log n)$. Their algorithm requires $O(\Delta + \log n)$ rounds and $O(n(\Delta + \log n))$ messages, where $\Delta$ is the maximum degree of any node in the initial graph. The model they used allows each node to send only one message per round to a neighbor, and the resulting binary tree can be further transformed into other desirable structures like expanders, butterflies, or hypercubes.  If the nodes are capable of sending and receiving an $O(\Delta \log n)$ number of messages per round, i.e., in \PCONGEST{} model, there exists a deterministic algorithm that operates in $O(\log^2 n)$ rounds, as shown in \cite{gmyr2017distributed}. Recently, this has been improved to $O(\log^{3/2} n)$ rounds with high probability, as demonstrated in \cite{gotte2019faster} for graphs with $\Delta$  polylogarithmic in $n$.

\citet{gilbert2020dconstructor} developed a different approach by designing a distributed protocol that efficiently reconfigures any connected network into a desired topology—such as an expander, hypercube, or Chord—with high probability. Here a node can send messages to all their neighbors in a round, regardless of their degree, resulting in faster communication for higher-degree nodes. Their protocol operates in $O(\polylog{n})$ rounds, utilizing messages of size $O(\log n)$ bits per link per round and achieving a message complexity of $\tilde{\Theta}(m)$. 

\citet{gotte2023time} later introduced an algorithm for constructing a well-formed tree—a rooted tree with constant degree and $O(\log n)$ diameter—from any connected graph. Their protocol first builds an $O(\log n)$-degree expander, which can be further refined into the desired tree structure. The algorithm is optimal in terms of time, completing in $O(\log n)$ rounds, which aligns with the theoretical lower bound of $\Omega(\log n)$ for constructing such topologies from arbitrary graphs~\cite{gotte2023time}. However, the message complexity remains $\tilde{\Theta}(m)$, as nodes are required to send and receive $d\log n$ messages per round, where $d$ is the initial maximum degree. The key innovation in their approach is the use of short random walks to systematically improve the conductance of the graph, ultimately leading to the formation of robust expander structures. Very recently, \cite{dufoulon2024time} introduces $\GOSSIPr{}(\log^2 n)$ model and showed that a constant-degree expander can be constructed starting from any initial topology by spending $O(\log^5 n)$ rounds and with message complexity $O(n \log^5 n)$. This algorithm in \cite{dufoulon2024time} is the first protocol that achieves message complexity independent of $m$. Note that our result on \GOSSIPr{} model (i.e., \Cref{cor:gossip}) is a strict improvement over the result of \cite{dufoulon2024time} in terms of both round and message complexity.

The $\HYBRID(\alpha, \beta, \gamma)$ model, initially introduced by \cite{augustine2020hybrid} for studying shortest paths, was further explored by \cite{gotte2023time}, who showed that in the $\HYBRID(\log n, 1, \log^2 n)$ model, an arbitrary topology can be transformed into a well-formed tree within $O(\log n)$ rounds. The message complexity of their algorithm is $O(m + n \log^3 n)$. In contrast, our result in the $\HYBRID(\log n, 1, \log n)$ model (\Cref{thm:hybridwft}) achieves communication efficiency, albeit in $O(\log^2 n)$ rounds. 

Research on overlay construction extends well beyond simple foundational examples, mainly due to the inherently dynamic nature of real-world overlay networks, which are often impacted by churn and adversarial behaviors. This research can be categorized into two primary areas: self-stabilizing overlays and synchronous overlay construction algorithms. Self-stabilizing overlays, which locally detect and correct invalid configurations, are extensively surveyed by Feuilloley et al.~\cite{feldmann2020survey}. However, many of these algorithms lack definitive communication complexity bounds and provide limited guarantees for achieving polylogarithmic rounds~\cite{berns2013building, jacob2014skip}. On the other hand, synchronous overlay construction algorithms are designed to preserve the desired network topology despite the presence of randomized or adversarial disruptions, thereby ensuring efficient load balancing and generating unpredictable topologies under certain error conditions~\cite{augustine2015enabling, drees2016churn, augustine2018spartan, gotte2019faster}. A significant advancement in this area is made by Gilbert et al.~\cite{gilbert2020dconstructor}, who demonstrated how fast overlay construction can be maintained even in the presence of adversarial churn, assuming the network stays connected and stable for an adequate duration. Additionally, Augustine et al.~\cite{angluin2005fast} investigated graph realization problems, focusing on rapidly constructing graphs with specific degree distributions; however, their approach assumes the initial network is arranged as a line, which simplifies the task. The complexity of overlay construction increases when nodes have restricted communication capabilities, prompting research into the Node-Capacitated Clique (NCC) model, where each node can send and receive $O(\log n)$ messages per round~\cite{augustine2019ncc}. Within the NCC model, efficient algorithms have been developed for various local problems such as MIS, matching, coloring, BFS tree, and MST~\cite{augustine2019ncc}. Notably, Robinson~\cite{robinson2021fast} established that constructing constant stretch spanners within this model necessitates polynomial time. Similar complexities are encountered in hybrid network models that blend global overlay communication with traditional frameworks like \LOCAL\ and \CONGEST, where extensive communication abilities enable solving complex problems like APSP and SSSP effectively, though often with considerable local communication overheads \cite{augustine2020hybrid, chang2024universally, kuhn2020shortest, feldmann2020survey}.

\subsection{Organization}

In \Cref{sec:prelim}, we present the basic graph terminologies and tools.
In \Cref{sec:star}, we present our protocols in the \GOSSIPr{} model, proving \Cref{thm:star}.
In \Cref{sec:hybrid},  we present our protocols in the $\HYBRID$ model, proving \Cref{thm:hybridwft}.
In \Cref{app:model}, we discuss the tradeoffs between some complexity measures.
In \Cref{sec:expander}, we provide the technical details for constructing a constant-degree expander from a constant-degree overlay.
In \Cref{sec:nodewiseprior}, we analyze the overall and node-wise message complexities of existing protocols.
In \Cref{sec:tree-wft}, we provide the technical details for constructing a well-formed tree from a rooted tree. 


\section{Preliminaries}\label{sec:prelim}

A graph is defined as $G = (V, E)$, where 
$E \subseteq \binom{V}{2}$, as the edges are undirected.
 The graph does not allow self-loops or multi-edges. Let $n = |V|$ and $m = |E|$. The neighborhood of a vertex $v$ in $G$ is denoted as $N_G(v) := \{ u \in V \mid \{u, v\} \in E \}$, and the degree of a vertex $v$ in $G$ is defined as $\deg_G(v) := |N_G(v)|$. The maximum degree of the graph is represented as $\Delta(G) = \max_{v \in V} \deg_G(v)$. The distance $d(u, v)$ between any two nodes $u$ and $v$ is the number of edges in the shortest path between them. The diameter of a graph is the maximum distance between any two nodes in the graph. The set of connected components of $G$ is denoted as $\CC(G)$, and the number of connected components in $G$ is represented by $\cc(G)$. 

A star graph, denoted as $K_{1, n-1} = (V, E)$, has a distinguished node $v \in V$ such that an edge $e = \{u, w\} \in E$ exists if and only if $v \in e$. A tree $T$ is a connected acyclic graph. A rooted tree $T_v$ is a tree with a distinguished node $v$ serving as the root. The depth of a rooted tree $T_v$ is the maximum distance from the root $v$ to any other node in the tree. A \emph{well-formed tree} is defined as a rooted tree with a constant maximum degree and a depth of $O(\log n)$. 
A \emph{\satt} is defined as a rooted tree with $O(\log n)$ maximum degree and $O\left(\log n\right)$ depth. 

We assume that each node has an ID of length $O(\log n)$ bits. The ID of an edge is a concatenation of the node IDs with the smaller first. We use $\eta(T_x)$ to denote the maximum edge ID among all edges incident to nodes in $T_x$. 

\paragraph{Expander:} The volume of any subset $S \subseteq V$ is defined as $\vol(S) := \sum_{v \in S} \deg_G(v)$. The conductance of a subset $S \subseteq V$, where $|S| \neq 0$ and $|S| \neq |V|$, is given by 

$$\Phi_G(S) := \frac{|E(S, V \setminus S)|}{\min(\vol(S), \vol(V \setminus S))},$$

\noindent where $E(S, V \setminus S) := \left\{ \{u, v\} \in E \mid u \in S, v \in V \setminus S \right\}$ represents the set of edges between $S$ and its complement.

The conductance of the graph $G$ is defined as 

$$\Phi(G) := \min_{S \subseteq V, S \neq \emptyset, S \neq V} \Phi_G(S).$$

Informally, a graph is considered an expander if it has high conductance. The thresholds commonly used to define high conductance vary by context, including $1/n^{o(1)}$, $1/\polylog(n)$, and $1/O(1)$. In this paper, we define an expander as a graph with conductance of $1/O(1)$.


\paragraph{Broadcast-and-echo:} A basic distributed protocol to disseminate and gather information is \emph{broadcast-and-echo}. It is initiated by some node $x$ and messages are relayed in a BFS manner, with possible modifications to the messages down the broadcasting tree. Then this process reaches the leaves, leaf nodes echo with some messages back to their parents. Internal nodes  wait untill all the messages are gathered before sending a computed aggregated message to their parents. This process takes $O(D(T_x))$ rounds and $O(|T_x|)$ messages, where $T_x$ refers to the broadcasting tree in this process. 

More generally, in the \CONGEST{} model with bandwidth $B$ bits, a broadcast-and-echo initiated by $x$ in $T_x$ with a maximum message size of $S$ bits can be done in $O\left(\frac{S}{B}+D(T_x)\right)$ rounds and $O\left(\frac{S}{B}|T_x|\right)$ messages, via message pipelining.

\paragraph{Find any outgoing edge:} For any tree $T_x$ rooted at $x$, we call edges between $T_x$ and $V \setminus T_x$ outgoing. Linear sketch techniques used in previous works by \cite{ahn2012analyzing,jowhari2011tight, pandurangan2018fast, dufoulon2024time} of $O(\log^2 n)$ bits can be used to sample an outgoing edge with constant success probability. To save on message complexity, we instead use a subroutine from \cite{king2015construction} to find an arbitrary outgoing edge from $T_x$. 

At a high level, this protocol of \cite{king2015construction} uses similar observation to the well-known linear graph sketch that internal edges in a tree will contribute $0$ to the sum of degree, or XOR of edge IDs.  However, it breaks the well-known linear graph sketch~\cite{ahn2012analyzing} into two phases. First, it uses $O(\log n)$ bits to aggregate the parity of the number of edges that is hashed into each log-scale bracket (1,2,4,8, ...). Then, they show that with constant probability there is one log-scale bracket that has exactly one edge hashed to it. This step corresponds to guessing the suitable sampling probability for exactly one outgoing edge to be sampled. They finally spend another $O(\log n)$ bits to identify the identity of the edge by XORing the edge numbers that are in the identified bracket. This process takes four iterations of broadcast-and-echo with message size $O(\log n)$ bits.

For completeness, we describe the protocol \FindOut(x) initiated at node $x$, which returns an edge leaving $T_x$ with probability at least $1/16$. The version described here corresponds to FindAny-C(x) in \cite{king2015construction}. \FindOut(x) uses another protocol from \cite{king2015construction} \HPTestOut(x) that returns true with high probability if there is an edge leaving $T_x$ and false otherwise. \HPTestOut{} is always correct if true is returned and uses one broadcast-and-echo with message size $O(\log n)$ bits.

\vspace{5pt} \noindent $\FindOut(x)$:
\begin{enumerate}
    \item $x$ initiates $\HPTestOut(x)$ in $T_x$ and return $\emptyset$ if \HPTestOut{} returns false.
    \item Determine the identity of an edge with the following steps:
        \begin{enumerate}
            \item $x$ broadcasts a random pairwise independent hash function $h : [1,\eta(T_x)] \to [0,r]$ where $r=2^w > \sum_{v\in T_x} \deg(v)$ for some $w$ .
            \item each node $y$ hashes the ID of all edges incident to it and compute a $\log r$-bit binary vector $\vec{h}(y)$ such that $\vec{h}_i(y) := |\{e\mid y\in e \land h(e) < 2^i \}| \mod 2$. 
            \item The vector $\vec{h}(T) := \oplus_{y\in T} \vec{h}(y)$ is computed up the tree, in the broadcast-and-echo return to $x$. Then $x$ broadcasts $min = \min\{i \mid \vec{h}_i(T) = 1\}$.
            \item Each node $y$ computes $s(y) = \oplus\{e \mid y \in e  \land h(e) < 2^{min}\}$ and $s(T)=\oplus_{y\in T} s(y)$ is computed up the tree in the broadcast-and-echo and returned to $x$. Observe that if there is exactly one edge leaving $T_x$ with $h(e) < 2^{min}$, then $s(T)$ is its edge ID.
        \end{enumerate}
    \item $x$ can perform another broadcast-and-echo to check if $s(T_x)$ is indeed a valid edge ID and return $s(T)$ if the check succeeds and $\emptyset$ if the check fails.
\end{enumerate}

\begin{lemma}[\cite{king2015construction}]\label{lm:findoutgoing}
    If there is no edge leaving $T_x$, then \FindOut(x) return $\emptyset$. Otherwise, $\FindOut(x)$ returns an edge leaving $T_x$ with probability at least $1/16$, else it returns $\emptyset$. The algorithm uses worst-case $O(D(T_x))$ rounds  and $O(|T_x|)$ messages.
\end{lemma}

\section{Star overlay construction in the \texorpdfstring{$\GOSSIPr$}{\GOSSIPr} model}
\label{sec:star}


We begin by describing the high-level approach underlying our algorithm for the \GOSSIPr{} model. We draw inspiration from the following existing techniques.

\paragraph{Boruvka-style cluster merging with efficient inter-cluster edge selection:}{
    We use a Boruvka-style cluster merging approach used in many prior works \cite{dufoulon2024time, gmyr2017distributed, angluin2005fast} while maintaining a simple and useful invariant. We start with each node being a single cluster. In each iteration, we select inter-cluster edges and merge the clusters joined by these edges. By ensuring a constant factor reduction in the number of clusters in each iteration, the process terminates in $O(\log n)$ iterations. The challenge here is how to quickly select an outgoing edge effectively (effectiveness measured by small round or message complexity). We adopt \Cref{lm:findoutgoing} to find an outgoing edge efficiently. This was not previously used in the overlay network construction context and is more efficient than the linear graph sketching technique used by \cite{dufoulon2024time}.
}

\paragraph{Selective merging to overcome long chains:}{
    Overall our protocol works by sampling an inter-cluster edge from each cluster and merging clusters joined by sampled edges. Merging can be potentially slow due to the large diameter when the inter-cluster edges form a long chain connecting many clusters. Therefore, we break this chain via a simple coin-flipping technique, where each cluster flips a coin and will only accept a request if the coins of the requesting and requested clusters satisfy a specific condition. This symmetry-breaking technique is used extensively in many parallel and distributed works in problems such as parallel list ranking \cite{cole1986deterministic} and distributed graph connectivity \cite{gazit1991optimal}.
}


\paragraph{Faster and simpler merging:}{
    A key difference between our protocol and that of \cite{dufoulon2024time} is that we maintain a much simpler and useful invariant (maintaining a star in each cluster) that allows us to aggregate information in a cluster and perform merging among clusters much faster.
}


\paragraph{Star overlay construction protocol:}{
    We describe our protocol \MergeStar{} to construct a star overlay in the \GOSSIPrOne{} model which proves \Cref{thm:star}. 
    
    The algorithm proceeds in $O(\log n)$ Boruvka-style phases \whp. In each phase, a constant factor of clusters is merged into other clusters to form a cluster for the next phase with constant probability. The algorithm starts with each node being a cluster and maintains the following invariant: at the end of each phase, every node in each cluster $S$ agrees on a leader $l(S)$. This is true initially since each node can be the leader of its own cluster. In other words, this invariant implies that each cluster will keep a star overlay topology.

    Denote the given input topology as $G=(V, E)$. Let $G_i = (V, E_i)$ be the topology of the overlay network at the end of phase $i$. Let $G_0 = \left(V, \emptyset\right)$, \ie we start with each node being an isolated node in the overlay network. We refer to a connected component $C = \left(S, E_i\cap \binom{S}{2}\right) \in \CC(G_i)$ as a cluster in $G_i$.\footnote{Sometimes we also loosely refer to $S$ as the cluster. This should not cause any confusion since a cluster $C$ in $G_i$ is induced by $S$.} An outgoing edge from the cluster $C$ is an edge in the input graph $G=(V, E)$ connecting a node in $S$ to a node outside $S$ \ie $\Out(C) := E_G(S, V\backslash S) := \left\{\{u, v\} \in E\mid u \in S \text{ and } v\in V \backslash S \right\}$ is the set of outgoing edges of the cluster $C$. 
    
    Each phase consists of three steps. In phase $i$, we start with the overlay $G_{i-1}$.
    \begin{enumerate}[itemsep=0pt, leftmargin=2em]
        \item \textit{Sampling Step:} Each cluster $C = \left(S, E_i\cap \binom{S}{2}\right)$ finds an outgoing edge $e$ from $\Out(C)$ and a color $\chi(S) \in \{\Blue, \Red\}$, and then sends a \textit{merging request} containing the sampled color $\chi(S)$ to the external node in the sampled edge. 
        \item \textit{Grouping Step:} Clusters who received \textit{merging requests} decide on which clusters to merge with based on the color and reply either with an \textit{accepting message} or a \textit{rejecting message}. 
        The purpose of the $ \{\Blue, \Red\}$-coloring is to prevent long chains of \textit{merging requests} slowing down the \textit{Merging Step}.
        \item \textit{Merging Step:} Clusters that agree on merging will perform this step to merge the clusters. Each node in these clusters must agree on a new leader to maintain the invariant. 
    \end{enumerate}

    \paragraph{Sampling step:}
    In this step, each cluster $C = \left(S, E_i\cap \binom{S}{2}\right)$ needs to sample an outgoing edge uniformly at random with constant probability. It will take $\Omega(\Delta)$ rounds if we let each node check if each neighbor is in the same cluster. 
    We will use the \FindOut{} protocol to reduce communication. Each broadcast-and-echo is replaced by each leave in the star sending one request to the leader $l(S)$ for the broad-casted information. This exploits the replying property of \GOSSIPr{}. No random walk or PUSH-style information-spreading like \cite{dufoulon2024time} is needed. After the \FindOut{} protocol, the leader finds an outgoing edge with constant probability. Then it sends a \textit{merging request} along the sampled edge to the destination. Note that this step works correctly with probability $1/16$ due to \Cref{lm:findoutgoing}, as long as $cc(G_i)\geq 2$.
    
    


    
    Additionally, to facilitate the grouping step,  the leader  $l(S)$  independently and uniformly samples a color $\chi(S) \in \{\Blue, \Red\}$ for the cluster $S$ and sends the $1$-bit information along with the \textit{merging request}.

    \paragraph{Grouping step:}{ 
    The node $v\in S'$ that received the \textit{merging request} will reply with the leader $l(S')$. Then each leader will send the request to other cluster leaders. Since each cluster can only initiate one \textit{merging request}, there will be in total $c(G_{i})$ {merging requests}. Thus, there can be cycles or long chains in the graph, which can affect merging speed. Thus, we need to break these chains. We do this by making each cluster leader accept a request if and only if it is \Red{} and the requesting cluster is \Blue{}.
    }

    \paragraph{Merging step:}{
    Note that after the Grouping step, the \textit{merging requests} viewed as edges among clusters will form a star with the \Red{} clusters as the centers. We can identify the merged clusters in $G_{i+1}$ with the \Red{} clusters in $G_i$. Thus, We can maintain the invariant in each cluster in $G_{i+1}$ by letting the leader of the \Red{} clusters become the new leader of the merged clusters in $G_{i+1}$. 
    
    Each \Red{} leader will reply with its own identifier to the \Blue{} leaders. Each \Blue{} cluster leader then broadcasts this new leader identifier to the \Blue{} cluster members by replying to the cluster members' request. Actions taken by different nodes in the merging step are summarized in \Cref{tab:merge}.

    \begin{table}[htbp]
        \centering
        \begin{tabular}{l|@{~~}l}
            \hline
            Type in $G_i$ & Actions in the Merging Step\\
            \hline
            Cluster member & Send leader update requests to its leader in $G_i$.\\
            \hline
            \Red{} leader $u$ & Reply to leader update requests with its own identifier.\\
            \hline
            \Blue{} leader $v$ & 
                \begin{tabular}{@{}l}
                    Received acceptance decision from $u$;\\
                    $\begin{cases}    
                        \text{Reply to leader update requests with $\id(u)$} & \text{if $v$ is accepted}\\
                        \text{Reply to leader update requests with $\id(v)$} & \text{if $v$ is rejected}\\
                    \end{cases}$\\
                \end{tabular}\\
            \hline
        \end{tabular}
        \caption{Summary of actions of different nodes in the Merging Step}
        \label{tab:merge}
    \end{table}

    After this, each node in $G_{i+1}$ will agree on the same leader (\ie the \Red{} leader), thereby maintaining the invariant.
    }
}

\paragraph{Analysis:}{
    We bound the number of phases that the algorithm takes before it terminates \whp.

    \begin{lemma}\label{lem:reduct}
        \label{lm:constantFactor} Let $i \in \mathbb{N}$ such that $\cc(G_i)\geq 2$. We have $\EE[\cc(G_{i+1})] \leq \frac{63}{64} \cc(G_i)$. 
    \end{lemma}
    \begin{proof}
        First, observe that in one phase, each cluster has done the sampling step and the grouping step which can affect the number of clusters at the end of the phase. 

        Let $X_C$ be the indicator random variable for the event that either $C$ fails to sample an outgoing edge or the \textit{requesting message} initiated by $C$ is \textbf{rejected}, for each cluster $C\in \CC(G_i)$. As $cc(G_i)\geq 2$, the leader of each cluster $C \in \CC(G_i)$ find an outgoing edge from $C$ with probability at least $1/16$. Moreover, note that the requesting message from cluster $C$ will be accepted with probability $1/4$. So, we have 
        
        $$\EE[X_C] \leq 1 - \frac{1}{16}\frac{1}{4} = \frac{63}{64}.$$ By linearity of expectation, we have 
\[\EE\left[\cc(G_{i+1})\right]=\EE\left[\sum_{C\in \CC(G_i)} X_C\right] = \sum_{C\in \CC(G_i)} \EE\left[X_C\right] \leq \frac{63}{64}\cc(G_i).\qedhere \]


    \end{proof}

    \begin{lemma}
        \label{lm:termination}
        The protocol \MergeStar{} terminates in $O(\log n)$ phases \whp.
    \end{lemma}
    \begin{proof}


Let $t = c \log n$ for some sufficiently large constant $c$ that we will determine later. Our goal is to show that $\cc(G_t) = 1$ \whp, implying that the protocol terminates in $O(\log n)$ phases \whp. Define $Y_i := \cc(G_i) - 1$, where $i$ is a non-negative integer. Initially, we have $Y_0 = n - 1$, and the process terminates in phase $i$ when $Y_i = 0$. We aim to demonstrate that $Y_t = 0$ \whp. Observe that $Y_0, \ldots, Y_t$ form a sequence of random variables that take non-negative integer values.

To achieve this, it suffices to show that $\EE[Y_i] \leq \frac{63}{64} \EE[Y_{i-1}]$ for every $i \in \mathbb{N}$. We start by establishing that $\EE[Y_i \mid Y_{i-1}] \leq \frac{63}{64} Y_{i-1}$ for each $i \in \mathbb{N}$.

Consider the case when $Y_{i-1} \geq 1$, i.e., $\cc(G_{i-1}) \geq 2$. Applying \Cref{lem:reduct}, we have:
$$
\EE[Y_i \mid Y_{i-1}] = \EE[\cc(G_i) \mid \cc(G_{i-1}) = Y_{i-1} + 1] - 1 \leq \frac{63}{64}(Y_{i-1} + 1) - 1 \leq \frac{63}{64} Y_{i-1}.
$$
On the other hand, if $i$ is such that $Y_{i-1} = 0$, then $Y_i = 0$. Thus, $\EE[Y_i \mid Y_{i-1}] \leq \frac{63}{64} Y_{i-1}$ holds trivially.
Hence, we can conclude:
$$
\EE[Y_i] = \EE[\EE[Y_i \mid Y_{i-1}]] \leq \EE \left[ \frac{63}{64} Y_{i-1} \right]=\frac{63}{64}\EE [Y_{i-1}] \quad, \text{for all } i \in \mathbb{N}.
$$
This implies:
$$
\EE[Y_t] \leq \left(\frac{63}{64}\right)^t \EE[Y_0] = \left(\frac{63}{64}\right)^t \cdot (n - 1).
$$

Since $t = c \log n$, we choose $c$ to be sufficiently large such that $\EE[Y_t] \leq \frac{1}{\poly(n)}.$ Applying Markov's inequality, we have $\Pr[Y_t \geq 1] \leq \frac{1}{\poly(n)}$. Since $Y_t$ only takes non-negative integer values, it follows that $Y_t = 0$ holds  \whp. Therefore, we conclude that the protocol terminates in $O(\log n)$ phases \whp.
    \end{proof}

Now we are ready to prove \Cref{thm:star}.

\star*
    
    \begin{proof}
        The correctness of this algorithm is obvious since a leader is maintained and known to all nodes in the clusters after each phase. Once the algorithm terminates, there will be only one cluster and all nodes will agree on a single leader. Therefore, at the end of the algorithm, we have constructed a star overlay network.

        By \Cref{lm:termination}, we know that the algorithm terminates in $O(\log n)$ phases \whp. Now we only need to check that it takes $O(1)$ rounds in each phase in the \GOSSIPrOne{} model to conclude that the algorithm terminates in $O(\log n)$ rounds with high probability. To be exact, we need seven rounds (described from the point of view of a cluster leader): $4$ rounds to run \FindOut{}; $1$ round to send the \textit{requesting message} and receive the new leader; $1$ round to resend the \textit{requesting message} to the cluster leaders and receive an \textit{accepting message} or a \textit{rejecting message}; and $1$ round to distribute the new leader identifier to the old cluster members. During this process, every cluster member just keeps sending requests to their old leader for the identifier of the new leader. 

        For $\GOSSIPr(b)$, where $b\in O\left(\log n\right)$, we can simulate one round of the above process with $O\left(\frac{\log n}{b}\right)$ rounds and arrive at our conclusion. 
        The $O\left(n \log n \cdot\max\left(\frac{\log n}{b}, 1\right)\right)$ messages \whp total message complexity follows from the restriction of the model that at most $O(n)$ messages are sent in each round.
    \end{proof}
}

\section{Well-formed tree overlay construction in the \HYBRID{} model}
\label{sec:hybrid}

In this section, we will describe a protocol for \wft\ (WFT) construction in the $\HYBRID(\log n, 1,\log n)$ model. Due to \Cref{lm:wft-expander} adapted from results in \cite{gotte2023time, dufoulon2024time}, we can convert a \wft{} overlay to a constant-degree expander overlay via only global communications in the $\HYBRID(\log n, 1, \log n)$ model with an additive $O(\log^2 n)$ round complexity and additive $O(\log^2 n)$ messages for each node. 
Therefore, we will focus on describing the WFT construction protocol.

\subsection{Algorithm}\label{subsect:algo}

Now we describe our protocol \hybridWFT{} to build the \wft{} to prove \Cref{thm:hybridwft}. The overall structure of the algorithm is similar to that in \Cref{sec:star}.
At the start, each node is a cluster, \ie $G_0$ consists only of isolated nodes. The algorithm proceeds in phases where in each phase a constant fraction of the clusters are merged in expectation. At the end of $O(\log n)$ phases, there will be only one cluster (per connected component of the input graph) \whp. We maintain the following invariant in each cluster after each phase: the subgraph induced by each cluster is a \satt{} ($O(\log n)$-degree, $O\left(\log n\right)$-depth) that spans the cluster. More specifically, each node in the cluster knows its parent and children, as well as the root of the \satt{}. 

In each phase, there will be 3 major steps. Unless mentioned otherwise, all communications are done over the global channel. 

\paragraph{Sampling step:}{
    This step aims to sample an outgoing edge from each cluster. We run the protocol $\FindOut(r)$ from the root $r$ and find an outgoing edge (from the cluster) with constant probability. This takes $O(D(T_x))=O\left(\log |T_x|\right)$ rounds and $O(|T_x|)$ messages.
    Next, the root node samples a random color $\chi \in \{\Red, \Blue\}$. Finally, $r$ broadcasts the selected outgoing edge and the color in the cluster.
}

\paragraph{Grouping step:}{
    The selected node in each cluster will send a merging request along the selected outgoing edge. The merging request contains the color and the root identifier of the requesting cluster. Since we sample from the local edges, all merging requests will be over the local edges. Then each node receiving a request accepts the request if and only if the accepting cluster is \Red{} and the requesting cluster is \Blue{}. The accepting node sends an accepting message with the identifier of the accepting node via the local edge. 

    To improve the probability of low local communication, we will need to reduce the effective neighbors of \Blue{} nodes. We do this via the following procedure. Each node $v$ receiving a request will compute a matching over its rejected neighbors. Then $v$ will send the \textit{rejecting message} along with its matched \textit{regrouping cluster} to each rejected requesting node. 
    Then each rejected node will send a \textit{regrouping message} to its matched \textit{regrouping cluster} over the global network. The matched pairs will exchange their cluster leader identifier to decide on a new leader based on who has a larger identifier.

    \begin{lemma}
        \label{lm:forest}
        The cluster graph $H^C=(\CC(G_i), E^C)$ is a forest, where $A, B\in \CC(G_i)$ are adjacent in $H^C$ if and only if one accepts the other or one is matched with the other. Moreover, each tree in $H^C$ has a diameter at most $3$.
    \end{lemma}

    \begin{proof}
        First, we call the edges from the accepted request $E_a^C$ and the edges from the matching $E_m^C$. Then $E^C = E_a^C \cup E_m^C$. 
        Since each cluster only initiates one request and we only accept requests from \Blue{} to \Red{}, $E_a^C$ induces a forest. Otherwise, there will be a cycle which consists of clusters of alternating colors. Then each cluster will be accepting some requests since each cluster only sends one request. However, this is impossible due to our acceptance rule because each accepting cluster (\Red{}) will not be accepted, and each accepted cluster (\Blue{}) will not be accepting any cluster. Since $E_m^C$ comes from matching the rejected nodes, and all rejected clusters are not connected by accepting edges, $H^C$ is a forest.

        For the diameter, see that the connected components induced by $E_a^C$ are stars centered at a \Red{} cluster. This is because each \Blue{} cluster rejects all requests and each request from a \Red{} cluster is rejected. Then, we perform a case analysis for the components connected by $E_m^C$. Let $\{A,B\}\in E_m^C$. If $A, B$ are both rejected \Blue{} clusters, then this component has diameter $1$. If $A, B$ are both rejected \Red{} clusters, then this component has a diameter at most $3$. If $A$ is \Red{} and $B$ is \Blue{}, this component has diameter at most $2$.
    \end{proof}
    Define $H = (V, E(G_i)\cup E_a \cup E_m)$, where $E_a$ are edges from the \textit{accepting messages} and $E_m$ are edges from the \textit{regrouping messages}.
    We call each connected component in $H$ a grouped cluster in phase $i$. Note that each grouped cluster has agreed on a unique leader. We will now perform the merging step to transform each grouped cluster into a \satt{}. 

    To illustrate the grouping step clearly, we show a possible execution of the grouping step in \Cref{fig:hybrid_group}.
    
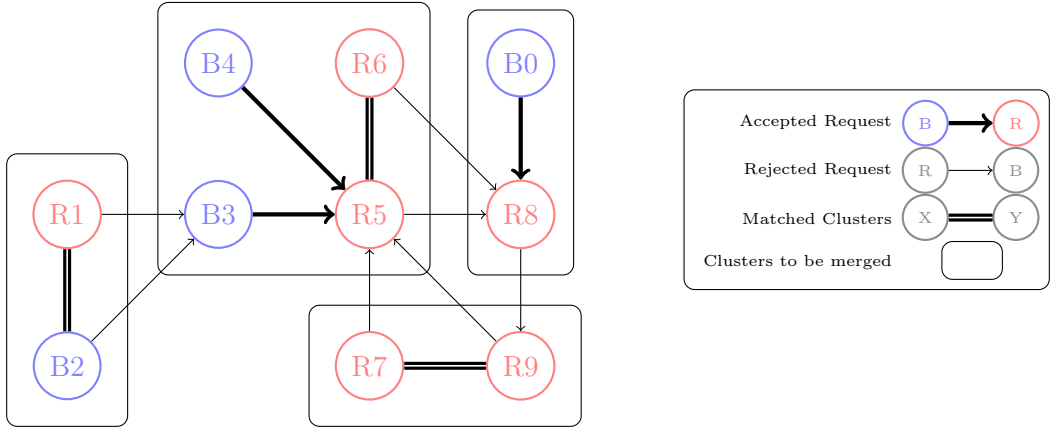
\begin{figure}[htbp]
    \centering
    \begin{tikzpicture}[node distance={20mm}, blue/.style = {draw, thick, circle, blue!50}, red/.style = {draw, thick, circle, red!50}, gray/.style = {draw, thick, circle, gray!90}]
        \node[red] (1) {R1};
        \node[blue] (2) [below of=1] {B2};
        \node[blue] (3) [right of=1] {B3};
        \node[blue] (4) [above of=3] {B4};
        \node[red] (5) [right of=3] {R5};
        \node[red] (6) [above of=5] {R6};
        \node[red] (7) [below of=5] {R7};
        \node[red] (8) [right of=5] {R8};
        \node[red] (9) [below of=8] {R9};
        \node[blue] (10) [above of=8] {B0};

        \draw[very thick, double] (1)--(2);
        \draw[->] (1)--(3);
        \draw[->] (2)--(3);
        
        \draw[ultra thick, ->] (3)--(5);
        \draw[ultra thick, ->] (4)--(5);
        \draw[ultra thick, ->] (10)--(8);
        \draw[very thick, double] (6)--(5);
        \draw[->] (7)--(5);
        \draw[->] (9)--(5);
        \draw[->] (5)--(8);
        \draw[->] (8)--(9);
        \draw[->] (6)--(8);
        \draw[very thick, double] (7)--(9);
        
        \draw[rounded corners] (-0.8, -2.8) rectangle (0.8, 0.8) {};
        \draw[rounded corners] (1.2, -0.8) rectangle (4.8, 2.8) {};
        \draw[rounded corners] (3.2, -2.8) rectangle (6.8, -1.2) {};
        \draw[rounded corners] (5.3, -0.8) rectangle (6.7, 2.7) {};

        \matrix [rounded corners, draw, above left, node distance={12mm}, font=\tiny] at (13,-1) {
            \node[] {Accepted Request};&
            \node[blue] (10) [] {B};
            \node[red] (11) [right of=10] {R};
            \draw[ultra thick, ->] (10)--(11); \\ \vspace{100pt}
            \node[] {Rejected Request};&
            \node[gray] (12) [] {R};
            \node[gray] (13) [right of=12] {B};
            \draw[->] (10)--(11); \\ 
            \node[] {Matched Clusters};&
            \node[gray] (14) [] {X};
            \node[gray] (15) [right of=14] {Y};
            \draw[very thick, double] (10)--(11); \\
            \node[] {Clusters to be merged};&
            \draw[rounded corners] (0,0) rectangle (0.8,0.5) {}; \\
        };
    \end{tikzpicture}
    \caption{A possible grouping step}
    \label{fig:hybrid_group}
\end{figure}
    
}

\paragraph{Merging step:}{
    Each node in a cluster that acts on behalf of the leader (to reply to requests) will inform its leader that it is the new leader of a grouped cluster.
    Each new leader will initiate a re-rooting process via a breadth-first search style broadcast, where each node will change its parent and children according to their distance from the new leader.
    Those requesting nodes that received an accepting message will re-root the requesting cluster towards the new leader by relaying this broadcast in the requesting cluster. 

    Let $v$ be the new leader of the grouped cluster. We now have a tree $T_v$ rooted at $v$ after re-rooting. $T_v$ has depth $O(\log n)$, since each original cluster has depth $O(\log n)$ and $H^C$ has diameter at most $3$ due to \Cref{lm:forest}. However, this tree might have a maximum degree up to $O(\Delta)$, due to the accepting edges \ie edges in $E_a$. 
    Therefore, we will perform the following transformation similar to the merging step in \cite{gmyr2017distributed} to maintain the invariant. However, we made some crucial adaptation to the pointer jumping step to reduce message complexity at the cost of introducing randomness.  
    First, we transform the tree into a child-sibling tree. Each node $v$ will sort its children in some arbitrary order and then attach itself at the head of this order. 
    Then for each child $u$ in this order, $v$ will send the previous and next node in the order. The last node will receive no next node. 
    In this way, each node keeps at most one child and one sibling. 
    By viewing the sibling as a child, we have constructed a binary tree. 
    Then, we can proceed with the same Euler Tour technique to turn this into a cycle of virtual nodes.
    Finally, we perform \RCtT{} described in \Cref{rc2t} to turn this cycle of virtual nodes into a tree with $O(\log n)$ degrees and $O(\log n)$ depth.


    The above steps to construct a \satt{} after re-rooting are described in more detail in \Cref{sec:tree-wft}, where we show that running this process for $O(\log n)$ times can be done in $O(\log^2 n)$ rounds with $O(n \log n)$ messages. Moreover, each node $v$ uses at most $O(\deg_G(v) + \log n)$ messages.
    
}

\paragraph{Post-processing:} After all the phases terminates, we get a cluster with a \satt{} overlay. We can now run one iteration of the deterministic \wft{} construction in \Cref{sec:tree-wft} to turn this \satt{} to a \wft{} with additive $O(\log n)$ rounds and $O(n\log n)$ messages.

\subsection{Analysis}
    
\paragraph{Round complexity:}{
    We first show that the process ends in $O(\log n)$ phases \whp, and then we are left to show that each phase takes $O(\log n)$ rounds. After that, we conclude that the whole algorithm terminates in $O\left(\log^2 n\right)$ rounds with an overlay network topology of a well-formed tree \whp. This is because the post processing takes only $O(\log n)$ rounds.

    \begin{lemma}
        \label{lm:hybrid termindation}
        The above protocol terminates in $O(\log n)$ phases \whp.
    \end{lemma}

    \begin{proof}
        The proof is similar to that of \Cref{lm:constantFactor} and \Cref{lm:termination}, as the additional matching step of rejected clusters only improves the cluster number reduction.
    \end{proof}
    
    \begin{lemma}
        The above algorithm takes $O(\log n)$ rounds in each phase.
    \end{lemma}

    \begin{proof}
        First, observe that all the following tasks take $O(\log n)$ rounds due to broadcasting $O(\log n)$ bits in a \satt{}: performing \FindOut{}, broadcasting the sampled outgoing edge and color, informing the leader that it is the leader of the grouped cluster and re-rooting the grouped cluster. 
    
        Then we look at the rest of the operations one by one. It takes $O(1)$ rounds to send and accept a request. Notably, accepting takes $O(1)$ rounds since accepting messages are sent over the local network and each node can send $O(1)$ message to each of its neighbors. Moreover, rejecting with \textit{regrouping clusters} takes $O(1)$ rounds for the same reason since the matching is computed locally. Comparing the new leader over the \textit{regrouping messages} also takes $O(1)$ round. 
        
        Transforming the grouped clusters into a child-sibling tree also takes $O(1)$ round. This is because, in the grouped clusters, each node only has $O(1)$ children in the overlay network that they have to contact over the global network. For the child-sibling tree transformation, only $O(1)$ message per child is needed. Therefore, each node only needs to send $O(1)$ messages which is below the global capacity. Over the local network, each node is allowed to send $O(1)$ messages to each of its neighbors, which is enough for our tasks. 
    
        The construction of the virtual Euler tour requires $O(1)$ rounds since the child-sibling tree is a constant degree tree. Lastly, \RCtT{} on the Euler tour takes $O(\log n)$ rounds to complete. 

        Overall, all steps are performed within $O(\log n)$ rounds and there are $O(1)$ steps in a phase. Therefore, each phase takes $O(\log n)$ rounds.
    \end{proof}
    
}

\paragraph{Message complexity:}{
    First, we will show that the total number of messages per phase is $O(n)$. Then the overall message complexity will be $O(n \log n)$ \whp since the protocol terminates in $O(\log n)$ phases \whp according to \Cref{lm:hybrid termindation} and the post-processing takes $O(n\log n)$ messages.

    \begin{lemma}
        \label{lm:hybrid-message-per-phase}
        Each phase in the above protocol takes $O(n)$ messages.
    \end{lemma}

    \begin{proof}
        First, we analyze the communication over the local channel. In each phase, local communication will only occur in the set of sampled outgoing edges $E_r$ whose size is bounded by $n$, \ie $|E_r|= |\CC(G_i)| \leq n$. Moreover, in each phase, there are only 3 steps where we make use of the local channel, namely, requesting, replying, and child-sibling tree transformation. In each of these steps, at most $O(1)$ messages will be sent in each edge in $E_r$. Therefore, there will be $O(n)$ messages sent in the local channel in each phase.
   
        To analyze the communication over the global channel, we first observe that constant rounds of broadcast-and-echo in \FindOut{} and other subroutines costs $O(n)$ messages. Next, it is shown in \Cref{rc2t} that the \RCtT{} subroutine which is run once in each phase, takes $O(n)$ messages. Lastly, the remaining tasks including child-sibling tree transformation and virtual Euler tour construction, take $O(1)$ messages per node. 
        
        In conclusion, $O(n)$ messages are sent in each phase.
    \end{proof}    
}

\paragraph{Node-wise message complexity:}{
    Node-wise message complexity refers to the maximum number of messages sent and received for a node over the whole execution of an algorithm. We claim that with high probability every node $v$ sends and receives at most $O(\deg_G(v) + \log n)$ messages in our protocol. 

\paragraph{Effective degree reduction:}
    To achieve low node-wise message complexity \whp, simply merging the cluster along sampled inter-cluster edges is not enough. We observe that we need to reduce the number of potential incoming requests to each node by a constant factor in each phase to reduce the node-wise message complexity. To achieve this, we further group and merge the \textit{requested but rejected} clusters of each node. Since each cluster can only send one request per phase, we have successfully reduced the ``effective degree'' of this node by reducing the number of potential requesting messages to this node in the future.

\paragraph{Intuition for analysis:}
    A node can send and receive messages from either the local network or the global network. The difficulty is to analyze the node-wise message complexity on the local network. 
    Intuitively, the number of local messages of a node is small due to two reasons. 
    First, if a request is accepted along an inter-cluster edge, the two clusters connected by this local edge will be merged and no more local communication is required for the remaining rounds of the algorithm's execution. 
    Second, local neighbors of a particular node $v$ can be merged into a cluster and hence there will be fewer possible requests sent from them since only one request can be sent from each cluster. 
    
    From these intuitions, we define the notion of an \emph{active neighboring cluster} of a node $v$. A cluster is said to be a neighbor of $v$ if there is a vertex in this cluster that is connected to $v$. We write $\mathcal{N}^i(v)$ to denote the set of all neighboring clusters of $v$ at the start of phase $i$. We call the cluster in $\mathcal{N}^i(v)$ containing $v$ the \textit{inactive} cluster and the other clusters the \textit{active} clusters, denoted as $\mathcal{N}^i_a(v)$. In the following lemma, replying messages refer to \textit{accepting messages} or \textit{rejecting messages} a node sent to \textit{merging requests}.
    
    \begin{lemma}
        \label{lm:node-wise-local}
        In the above protocol, any node $v$ sends a total $O(\deg(v)+\log n)$ replying messages over the local network before the algorithm terminates \whp.
    \end{lemma}
    \begin{proof}
        We assume that the inter-cluster edge sampling behavior is controlled by an oblivious adversary that does not know the random outcome of the color sampling in each cluster.

        Suppose at the start of the $i$-th phase, the set of active neighboring clusters of $v$ is $\mathcal{N}^i_a(v)$. Suppose further that $k_i$ active clusters send requests to $v$ in phase $i$, where $k_i \leq |\mathcal{N}^i_a(v)|$. 
        Since the algorithm terminates in $O(\log n)$ phases \whp by \Cref{lm:hybrid termindation}, and we are aiming for an $O(\deg(v)+\log n)$ term, we call any phases with $k_i\leq 1$ a trivial phase, and other phases non-trivial phases. Note that $v$ replies to a total of no more than $O(\log n)$ requests in all the trivial phases \whp. Now we will assume that $k_i > 1$ and try to bound the number of replying messages in non-trivial phases.
        
        Let $Y_i=|\mathcal{N}^i_a(v)|$ be the number of active neighboring clusters of $v$ at the start of phase $i$. 
        If $v$ is \Blue{}, all of the $k_i$ requesting clusters will be rejected. Due to the matching mechanism, the number of active neighboring clusters of $v$ will decrease by $\floor{\frac{k_i}{2}}$. 
        If $v$ is \Red{}, we will accept the \Blue{} requests and pair up the rejected requests. Suppose $p$-fraction of the requests are rejected and $(1-p)$-fraction of the requests are accepted, then the number of active neighboring clusters of $v$ will decrease by $\floor{\frac{p k_i}{2}} + (1-p)k_i\geq \floor{\frac{k_i}{2}}$. Since $k_i > 1$, in any non-trivial phase $i$, $Y_{i+1} \leq Y_i - \floor{\frac{k_i}{2}} \leq Y_i - \frac{k_i}{4}$.
        

        Suppose the algorithm terminates in $C\log n$ phases \whp, for some constant $C$. Define $T\subseteq [C\log n]$ as the set of trivial phases and $S=[C\log n]\backslash T$ the set of non-trivial phases. 
        For each $i\in S$, $Y_{i+1} \leq Y_{i} - \frac{1}{4}k_i$ implies that $k_i \leq 4(Y_{i} - Y_{i+1})$.
        The total number of replies sent by $v$ is as follows.
        \begin{align*}
            \sum_{i=1}^{C\log n} k_i 
            &= \sum_{i\in S} k_i + \sum_{i\in T} k_i\\
            &\leq \sum_{i \in S} 4(Y_{i} - Y_{i+1}) + O(\log n)\\
            &\leq \sum_{i \in [C\log n]} 4(Y_{i} - Y_{i+1}) + O(\log n) \tag{Since $Y_{i} \geq Y_{i+1}$}\\
            &\leq 4Y_0 +O(\log n) = O(\deg(v)+\log n). && \qedhere
        \end{align*}
    \end{proof}

    \begin{restatable}{lemma}{node}
        \label{lm:node-wise-message}
        In the above protocol, any node $v$ sends a total $O(\deg_G(v)+\log n)$ messages before the algorithm terminates \whp.
    \end{restatable}
    \begin{proof}[Proof of \Cref{lm:node-wise-message}]

        First, we analyze the communication over the global channel. Each node only participates in $O(1)$ calls of broadcast-and-echo in each phase. Since the node-wise message complexity of broadcast-and-echo is proportional to its degree in a phase, we first analyze a node's degree in a phase. Although in each \satt{}, each node can have $O(\log n)$ degree in the worst case, we now show that the sum of the nodes degree across $O(\log n)$ phases is at most $O(\log n)$. It is shown in \Cref{rc2t} that each node is active for $O(\log n)$ coin flip iterations during $O(\log n)$ runs of \RCtT{}. Since participating in each coin flip iteration contribute to constant number of messages and constant number of child, we can conclude that each node sends a total $O(\log n)$ messages (for both construction of \RCtT{} messages sent for broadcast-and-echo, which is proportional to the sum of degrees) across $O(\log n)$ phases \whp. 
        
        Next, we analyze the communication in the local channel. We only need to consider two cases: (1) $v$ sends a request and is rejected, and (2) $v$ replies (either reject or accept) other requests sent to it via active local edges. The first case can happen $O(\log n)$ times since the algorithm terminates in $O(\log n)$ phases \whp according to \Cref{lm:hybrid termindation}, giving rise to $O(\log n)$ messages, which is dominated by the  $O(\log n)$ term. 

        The number of messages $v$ sends in case (2) is characterized by \Cref{lm:node-wise-local}. Hence we can arrive at our conclusion.
    \end{proof}

    Now we are ready to prove \Cref{thm:hybridwft}.

    \hybridwft*
    \begin{proof}
       The theorem follows from \Cref{lm:hybrid termindation}, \Cref{lm:hybrid-message-per-phase}, and \Cref{lm:node-wise-message}.
    \end{proof}
}




\section*{Acknowledgments}
We are grateful to an anonymous FSTTCS reviewer for suggesting the use of the hash functions of~\cite{king2015construction} to reduce the number of bits required for identifying an outgoing edge. In an earlier version of this work, this task had been carried out using the linear sketches of~\cite{ahn2012analyzing}. The incorporation of the reviewer’s suggestion results in an $O(\log n)$-factor improvement in the overall communication complexity.

\printbibliography


\appendix

\section{Tradeoffs between complexity measures}\label{app:model}

{
There is a tradeoff among round complexity, message complexity, and balanced communication: If each node in $G$ is only allowed to send $O(1)$ messages per round, then any algorithm to build a constant-degree overlay $H$ requires $\Omega(\Delta(G))$ rounds~\cite{angluin2005fast}. To achieve both low round complexity and low message complexity independent of the initial degree of the input graph, certain nodes may need to send many more messages than others. This is observed in our first protocol in the \GOSSIPr{} model in \Cref{sec:star}. Although it achieves good round and message complexities, it suffers from high regional communication where some nodes may need to communicate up to $\Omega(n)$ messages. 

\paragraph{Node-wise message complexity:}
Node-wise message complexity can be seen as a measure that quantifies the aforementioned imbalance. 
The study of node-wise complexity is further motivated by the pursuit of fairness in the P2P network context. An extensive body of work is devoted to designing fair mechanisms to encourage users to contribute to the P2P network~\cite{anceaume2005incentives, lakhani2022fair}. However, besides offering incentives, it is essential to guarantee that users will incur low and fair costs when they join the network. A node may be discouraged from joining the network if it may potentially be selected as a crucial node of the network and perform much more work than other participating nodes.
}

\paragraph{Communication capacity versus round complexity:}{
    Communication capacity refers to the number of messages sent per node per round. This parameter captures the congestion in real-world networks. Intuitively, algorithms designed with more stringent communication capacity can perform better in real-world networks under congestion. There has been a series of works on improving the round complexity of the overlay network construction problem with more stringent communication capacity. Specifically, \citeauthor{angluin2005fast}~\cite{angluin2005fast} asked if there is an $O(\log n)$-round algorithm with $O(\Delta)$ communication capacity. While \citeauthor{gotte2023time}~\cite{gotte2023time} answered this question affirmatively for the case that the communication capacity is $O\left(\Delta+\log^2 n\right)$, in this work we present another tradeoff with $O\left(\log^2 n\right)$ round complexity and $O(\Delta+\log n)$ communication capacity, see \Cref{thm:hybridwft} and \Cref{cor:hybridwf}. 
}


\section{Expander construction in constant degree network}\label{sec:expander}


\wfttoexp*
\begin{proof}
    Firstly, we will apply the \CExp{} procedure from \cite{gotte2023time} on the global network which uses $O(\log n)$ rounds and $O(n\log^2 n)$ messages to convert a constant degree network $G$ into an $O(\log n)$-regular expander network $H$ with conductance $\phi \in \left(0, \frac{1}{2}\right]$ \whp. Note that the round complexity in our $\HYBRID(\log n, 1, \log n)$ model follows since we started with a constant degree graph and we allow $O(\log n)$ messages to be sent. This aligns with the parameters in \cite{gotte2023time}. The node-wise message complexity and the message complexity follow by multiplying a $\log n$ factor since each node is allowed to send $O(\log n)$ messages per round.

    Then, we can apply the \ExpDR{} procedure on the global network from \cite{dufoulon2024time} that can convert any $O(\log n)$-regular expander graph $H$ with constant conductance into a bounded-degree expander graph $H_1$ with conductance $\Phi \in \left(0,{1}/{10}\right]$ \whp. 
    
    Roughly speaking, in the \ExpDR{} procedure, each node generates $c=O(1)$ active tokens. Then in each of the $O(\log n)$ phases, each active token performs $O(\log n)$ steps of random walks. At the end of a phase, if a token ends in a node with less than $10c$ tokens, then the token stays in that node, becomes inactive, and informs the source about the inactivity and the destination of this token. Otherwise, the token will perform another $O(\log n)$-step random walk in the next phase. Once all tokens become inactive, all nodes establish a link with every other node holding its inactive tokens. See \cite{dufoulon2024time} for the proof.

    Since each node started with $c = O(1)$ tokens at the start and the graph is regular, the expected number of tokens at each node $v$ at each step of the random walk is also $c = O(1)$ throughout the entire walk. By Chernoff bound, we know that each node will receive at most $O\left(\frac{\log n}{\log \log n}\right)$ tokens in each round \whp. Therefore, the $O(\log^2 n)$ steps of random walk for each token can be done in $O(\log^2 n)$ rounds. The node-wise complexity follows by multiplying the number of rounds by the load \whp. The algorithm only uses $O(n\log^2 n)$ messages since we only generate $c n$ tokens and each token performs at most $O(\log^2 n)$ steps of random walk.  
\end{proof}

\section{Message complexities in prior work}\label{sec:nodewiseprior}

In this section, we analyze the overall and node-wise message complexities of the two protocols in \cite{gotte2023time}.

\paragraph{Algorithm in \PCONGEST{}:}
Recall that in the \PCONGEST{} model each node $v$ can send and receive at most $\deg(v)\log n$ messages. The protocol consists of $O(\log n)$ iterations. In each iteration, each node $v$ creating $O\left(\deg_G(v)\log n\right)$ tokens that do constant-length random walk. 
Since the expected degree-normalized load for each node never increases, by Chernoff bound, each node receives at most $O\left(\deg_G(v)\cdot \frac{\log^2 n}{\log \log n}\right)$ tokens in each round. Hence, the node-wise message complexity is $O\left(\deg_G(v)\cdot \frac{\log^3 n}{\log \log n}\right)$.
The total message complexity is $\Theta(m\log^2 n)$, since in each phase, each node generates $O(\log n)$ constant-length random walk for each of its neighbors.

\paragraph{Algorithm in $\mbox{\HYBRID}(\log n, 1, \log^2 n)$:}

The algorithm first transforms the input graph into an $O(\log n)$-degree graph via the local network and performs their expander creation algorithm described in the previous paragraph in the global network, which is equivalent to the $\NCC_0$ or the \PCONGEST{} model with maximum degree $O(\log n)$. Each node $v$ uses $O(\deg_G(v))$ messages in the transformation step and at most $O\left(\log n \cdot \frac{\log^3 n}{\log \log n}\right)=O\left(\frac{\log^4 n}{\log \log n}\right)$ messages in the second step.
The total message complexity is $\Omega(m + n\log^3 n)$, since it takes $\Omega(m)$ messages for the first step and the transformed graph has $\Theta(n\log n)$ edges.

\section{Converting a rooted tree into a \wft{}}
\label{sec:tree-wft}
In this section, we provide the details of converting the rooted tree overlay of $n$ nodes generated by the grouping step of our \hybridWFT{} protocol into a \wft{} in \Cref{subsect:algo}.

We first describe a deterministic protocol to convert a rooted tree into a \wft{}, since that is the target topology at the end of the algorithm. This method is relatively folklore and a similar description can be found in \cite{gmyr2017distributed}. Additionally, we describe a modification of this method with randomness to improve message complexity. This is applied after the broadcasting-style re-rooting step to maintain the invariant that each cluster is a \wft{}.

\subsection{Deterministically converting a rooted tree into a \wft{}}\label{sec:detwft}

\paragraph{Goal:} We will demonstrate a {deterministic} protocol in the $\HYBRID(\log n, 1, \log n)$ model that constructs a well-formed tree overlay from the rooted tree overlay of $n$ nodes generated by the grouping step of our \hybridWFT{} protocol described in \Cref{subsect:algo}. The protocol costs $O(\log n)$ rounds and $O(n \log n)$ messages. Moreover, each node $v$ sends and receives at most $O(\deg_G(v) + \log n)$ messages. We first describe the protocol in three steps and then show an analysis. 

    \paragraph{Child-sibling tree transformation:}{
        Each node will compute an ordering of its children locally. For example, suppose node $v$ has children $u_1, u_2, \dots, u_k$ in order, where $k$ is the number of children that $v$ has in the rooted tree.\footnote{Note that $k\leq \max(\deg_G(v), c)$, where $\deg_G(v)$ denotes the degree of $v$ in the local graph and $c$ is the constant upper bound for the degree of the \wft{} from the previous iteration.}
        Then to each node $u_i\in \{u_0,\dots, u_k\}$, $v$ will send the identifier of $u_{i-1}$, with the convention that $u_0 := v$, to indicate the new ``parent/sibling'' of $u_i$. This step is illustrated pictorially in \Cref{fig:child_sib}. 
        Now, if we view each incoming arrow as a parent, each node now has at most $2$ children and $1$ parent.  
        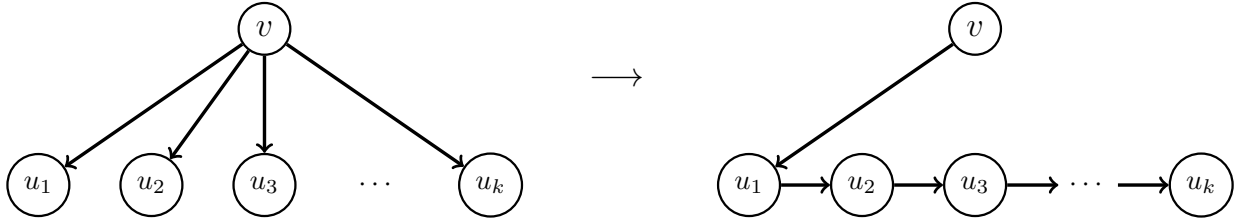
\begin{figure}[htbp]
    \centering
    \resizebox{\linewidth}{!}{
    \begin{tikzpicture}[node distance={14mm}, font=\small, main/.style = {draw, thick, circle}]
        \node[main] (1) {$u_1$};
        \node[main] (2) [right of=1] {$u_2$};
        \node[main] (3) [right of=2] {$u_3$};
        \node (4) [right of=3] {$\cdots$};
        \node[main] (5) [right of=4] {$u_k$};
        \node[main, font=\large] (6) [above=12mm of 3] {$v$};

        \node[font=\large] (A) [above right=8mm and 8mm of 5] {$\longrightarrow$};
        
        \node[main] (11) [right=24mm of 5] {$u_1$};
        \node[main] (12) [right of=11] {$u_2$};
        \node[main] (13) [right of=12] {$u_3$};
        \node (14) [right of=13] {$\cdots$};
        \node[main] (15) [right of=14] {$u_k$};
        \node[main, font=\large] (16) [above=12mm of 13] {$v$};

        \draw[very thick, ->] (6)--(1);
        \draw[very thick, ->] (6)--(2);
        \draw[very thick, ->] (6)--(3);
        \draw[very thick, ->] (6)--(5);
        \draw[very thick, ->] (16)--(11);
        \draw[very thick, ->] (11)--(12);
        \draw[very thick, ->] (12)--(13);
        \draw[very thick, ->] (13)--(14);
        \draw[very thick, ->] (14)--(15);
    \end{tikzpicture}
    }
    \caption{An illustration of the child-sibling tree transformation at node $v$}
    \label{fig:child_sib}
\end{figure}
    }

    \paragraph{Euler tour technique:}{
        Each node $v$ internally simulates $k$ virtual nodes, $v_1, \dots, v_k$, where $k\leq 3$. Then $v$ informs its $k$ neighbors of their $2$ adjacent virtual nodes. The adjacency of the virtual nodes is consistent with an in-order traversal of the rooted tree. Since each node is of degree $3$ after the child-sibling tree transformation, the in-order traversal will visit each node at most $3$ times, thus creating at most $3$ virtual nodes.
        This step is illustrated in \Cref{fig:euler}. After we complete this step, we will have a cycle of virtual nodes where each node has at most $3$ virtual copies in this cycle.
        \begin{figure}[htbp]
    \centering
    \begin{tikzpicture}[node distance={18mm}, font=\footnotesize, main/.style = {draw, thick, circle, inner sep=0.6mm}]
        \node[font=\large, main,inner sep=1.2mm] (1) {$v$};
        \node[font=\large, main,inner sep=1.2mm] (2) [above right of=1] {$u$};
        \node[font=\large, main,inner sep=1.2mm] (3) [below right of=1] {$w$};
        \node[font=\large, main,inner sep=1.2mm] (4) [below left of=1] {$x$};

        \node[font=\large] (A) [right=30mm of 1] {$\longrightarrow$};
        
        \node[main] (11) [right=26mm of A] {$v_1$};
        \node[main] (112) [right=3.4mm of 11] {$v_2$};
        \node[main] (113) [below right=0.8mm and 0.2mm of 11] {$v_3$};
        \node[main] (12) [above right of=11] {$u_1$};
        \node[main] (122) [right=3.4mm of 12] {$u_2$};
        \node[main] (123) [below right=0.8mm and 0.2mm of 12] {$u_3$};
        
        \node[main] (13) [below right of=11] {$w_1$};
        \node[main] (132) [right=3.4mm of 13] {$w_2$};
        \node[main] (133) [below right=0.8mm and 0.2mm of 13] {$w_3$};
        
        \node[main] (14) [below left of=11] {$x_1$};
        \node[main] (142) [right=3.4mm of 14] {$x_2$};
        \node[main] (143) [below right=0.8mm and 0.2mm of 14] {$x_3$};

        \draw[very thick, ->] (1)--(3);
        \draw[very thick, ->] (1)--(4);
        \draw[very thick, ->] (2)--(1);

        \draw[thick, <->] (12)--(11);
        \draw[thick, <->] (11)--(14);
        \draw[thick, <->] (142)--(113);
        \draw[thick, <->] (113)--(13);
        \draw[thick, <->] (132)--(112);
        \draw[thick, <->] (112)--(123);
    \end{tikzpicture}
    \caption{An illustration of the Euler Tour simulation at node $v$ (other parts of the graph omitted)}
    \label{fig:euler}
\end{figure}
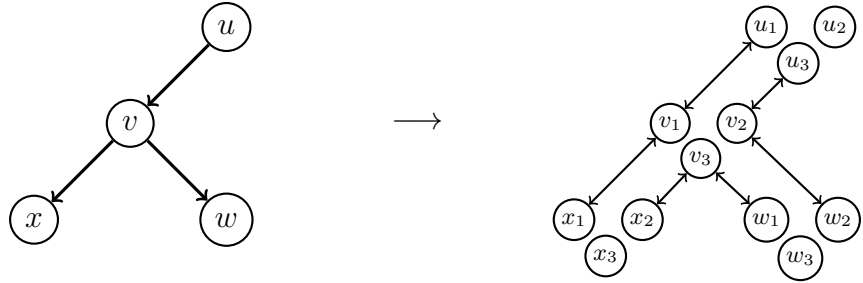

    }

    \paragraph{Deterministic pointer jumping:}{
        Each node will perform the following for $\log K$ iteration, where $K$ is the number of virtual nodes (this number can be determined exactly by performing a broadcast and converge-cast from one of the virtual nodes of the root $r$). We denote $N_i(v)$ to be the neighborhood $v$ learned in iteration $i$, with $N_0(v)$ being the initial Euler Tour neighborhood. In iteration $t$, each virtual node will introduce its two neighbors in $N_{t-1}(v)$ to each other. Thus, in iteration $t$, each virtual node will learn about the two virtual nodes that are of distance $2^{t}$. Then, we can start building the well-formed tree. One of the virtual nodes of root $r$ will start broadcasting along its latest learned neighbors $N_{\floor{\log K}-1}(r)$. Then, in iteration $t$, each node $v$ will broadcast along neighbour $N_{\max(\floor{\log K}-1-t, 0)}(v)$. It is easy to show that this process covers all the virtual nodes. After this process is done, we keep only these broadcasting edges. We can see that these broadcasting edges form a binary tree of depth $O(\log K) = O(\log n)$. 

        \def\rh{11}
\def\rd{5.8}
\def\hh{13}

\begin{figure}[htbp]
    \centering
    \begin{tikzpicture}[node distance={\rh mm}, main/.style = {draw, circle,inner sep=0.8mm}]
        \node[] (0) {};
        \node[main] (1) [right of=0] {};
        \node[main] (2) [above right=\rd mm and \rd mm of 0] {};
        \node[main] (3) [above of=0] {};
        \node[main] (4) [above left=\rd mm and \rd mm of 0] {};
        \node[main] (5) [left of=0] {};
        \node[main] (6) [below left=\rd mm and \rd mm of 0] {};
        \node[main] (7) [below of=0] {};
        \node[main] (8) [below right=\rd mm and \rd mm of 0] {};

        \draw[very thick] (1) to [out=90,in=315] (2);
        \draw[very thick] (3) to [out=0,in=135] (2);
        \draw[very thick] (3) to [out=180,in=45] (4);
        \draw[very thick] (5) to [out=90,in=225] (4);
        \draw[very thick] (5) to [out=270,in=135] (6);
        \draw[very thick] (7) to [out=180,in=315] (6);
        \draw[very thick] (7) to [out=0,in=225] (8);
        \draw[very thick] (1) to [out=270,in=45] (8);

        \node[font=\large] (A) [right=\hh mm of 0] {$\longrightarrow$};

        \node[] [right=\hh mm of A] (10) {};
        \node[main] (11) [right of=10] {};
        \node[main] (12) [above right=\rd mm and \rd mm of 10] {};
        \node[main] (13) [above of=10] {};
        \node[main] (14) [above left=\rd mm and \rd mm of 10] {};
        \node[main] (15) [left of=10] {};
        \node[main] (16) [below left=\rd mm and \rd mm of 10] {};
        \node[main] (17) [below of=10] {};
        \node[main] (18) [below right=\rd mm and \rd mm of 10] {};

        \draw[] (11) to [out=90,in=315] (12);
        \draw[] (13) to [out=0,in=135] (12);
        \draw[] (13) to [out=180,in=45] (14);
        \draw[] (15) to [out=90,in=225] (14);
        \draw[] (15) to [out=270,in=135] (16);
        \draw[] (17) to [out=180,in=315] (16);
        \draw[] (17) to [out=0,in=225] (18);
        \draw[] (11) to [out=270,in=45] (18);
        \draw[very thick] (11)--(13);
        \draw[very thick] (13)--(15);
        \draw[very thick] (15)--(17);
        \draw[very thick] (17)--(11);
        \draw[very thick] (12)--(14);
        \draw[very thick] (14)--(16);
        \draw[very thick] (16)--(18);
        \draw[very thick] (12)--(18);

        \node[font=\large] (B) [right=\hh mm of 10] {$\longrightarrow$};

        \node[] [right=\hh mm of B] (20) {};
        \node[main] (21) [right of=20] {};
        \node[main] (22) [above right=\rd mm and \rd mm of 20] {};
        \node[main] (23) [above of=20] {};
        \node[main] (24) [above left=\rd mm and \rd mm of 20] {};
        \node[main] (25) [left of=20] {};
        \node[main] (26) [below left=\rd mm and \rd mm of 20] {};
        \node[main] (27) [below of=20] {};
        \node[main] (28) [below right=\rd mm and \rd mm of 20] {};

        \draw[] (21) to [out=90,in=315] (22);
        \draw[] (23) to [out=0,in=135] (22);
        \draw[] (23) to [out=180,in=45] (24);
        \draw[] (25) to [out=90,in=225] (24);
        \draw[] (25) to [out=270,in=135] (26);
        \draw[] (27) to [out=180,in=315] (26);
        \draw[] (27) to [out=0,in=225] (28);
        \draw[] (21) to [out=270,in=45] (28);
        \draw[] (21)--(23);
        \draw[] (23)--(25);
        \draw[] (25)--(27);
        \draw[] (27)--(21);
        \draw[] (22)--(24);
        \draw[] (24)--(26);
        \draw[] (26)--(28);
        \draw[] (22)--(28);
        \draw[very thick] (21) to [out=170,in=10] (25);
        \draw[very thick] (22) to [out=235,in=35] (26);
        \draw[very thick] (23) to [out=280,in=80] (27);
        \draw[very thick] (24) to [out=325,in=125] (28);
        \draw[very thick] (21) to [out=190,in=350] (25);
        \draw[very thick] (22) to [out=215,in=55] (26);
        \draw[very thick] (23) to [out=260,in=100] (27);
        \draw[very thick] (24) to [out=305,in=145] (28);

        \node[font=\large] (C) [right=\hh mm of 20] {$\longrightarrow$};

        \node[] [right=\hh mm of C] (30) {};
        \node[main] (31) [right of=30] {};
        \node[main] (32) [above right=\rd mm and \rd mm of 30] {};
        \node[main] (33) [above of=30] {};
        \node[main] (34) [above left=\rd mm and \rd mm of 30] {};
        \node[main] (35) [left of=30] {};
        \node[main] (36) [below left=\rd mm and \rd mm of 30] {};
        \node[main] (37) [below of=30] {};
        \node[main] (38) [below right=\rd mm and \rd mm of 30] {};

        \draw[] (33) to [out=0,in=135] (32);
        \draw[] (33) to [out=180,in=45] (34);
        \draw[] (37) to [out=180,in=315] (36);
        \draw[] (37) to [out=0,in=225] (38);
        \draw[] (33)--(35);
        \draw[] (35)--(37);
        \draw[] (31) to (35);

    \end{tikzpicture}
    \caption{An illustration of the pointer jumping step with $8$ virtual nodes}
    \label{fig:pointer}
\end{figure}
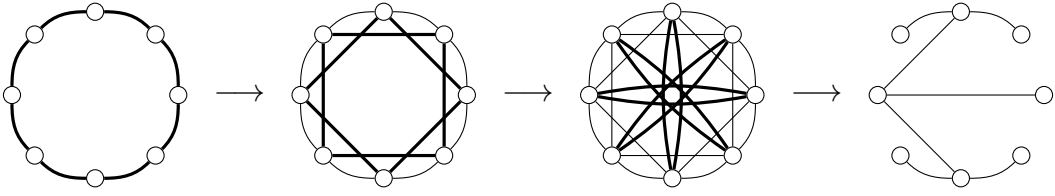
    }

    \paragraph{Post-processing:} After the previous step, we get a binary tree on the virtual nodes. Finally, we merge the virtual nodes and we get a graph with degree at most $6$ and diameter at most $O(\log n)$, since merging does not increase the diameter. Since this resultant may not be a tree, we will perform a breadth-first search from the root $r$ to finish our construction of the \wft{}.

    \paragraph{Analysis:}{
        First, we analyze the round complexity. The child-sibling tree transformation only takes $1$ round since each node is only adjacent to a constant number of children over the global network due to the maintenance of the invariant in the previous phrase and each node is allowed to send $1$ message to each of its local neighbors. It is also clear that the Euler tour simulation step only takes $O(1)$ rounds since each node only needs to share a constant number of messages with each of its constant number of neighbors. The pointer jumping step takes $O(\log n)$ rounds and the breadth-first search style broadcasting takes $O(\log n)$ steps.

        The child-sibling tree transformation takes at most $O(n)$ messages since each edge is used once and the graph before the transformation is a tree. The Euler Tour simulation also takes at most $O(n)$ messages since each node only sends a constant number of messages to each node in its neighborhood, which is bounded by $3$.
        Each node sends at most $O(\log n)$ messages for the pointer jumping steps since each node sends $O(1)$ messages per iteration for $O(\log n)$ iterations. Therefore, the above three steps use at most $O(n\log n)$ messages.

        Lastly, each node $v$ sends at most $O(\deg_G(v))$ messages over the local edges (for the child-sibling transformation step) and at most $O(\log n)$ messages over the global edges. Hence, the above steps have a node-wise complexity of $O(\deg_G(v)+\log n)$.}

\subsection{Message-efficient adaptation }
\label{rc2t}
\paragraph{Goal:} We will demonstrate a {randomized} protocol in the $\HYBRID(\log n, 1, \log n)$ model that constructs a \emph{\satt{}{}} overlay from the rooted tree overlay of $n$ nodes generated by the grouping step of our \hybridWFT{} protocol described in \Cref{subsect:algo}. 
We show that running $O(\log n)$ repetitions of this protocol costs $O(\log^2 n)$ rounds and $O(n \log n)$ messages \whp. Moreover, each node $v$ sends and receives at most $O(\deg_G(v) + \log n)$ messages over $O(\log n)$ repetitions.  

The first two steps, Child-sibling tree transformation and Euler tour transformation, are identical to \Cref{sec:detwft}. We now describe a randomized protocol to turn a cycle into a \satt{} with $O(\log n)$ rounds and $O(n)$ messages \whp. This step replaces the deterministic pointer jumping step.

\paragraph{Randomized \satt{}{} construction} We call this process \RCtT for ``randomized cycle to tree''. We start with a cycle of virtual nodes. Each node repeatedly performs the following steps:

\begin{enumerate}
    \item Each active node flip a fair coin and sends its coin value along with the IDs of its current two neighbors to its two neighbors.
    \item Each node will decide to become inactive if it is head and both of its neighbors are tail. In this case, it will send a message to both neighbors to declare its decision to be inactive, with a message to an arbitrary tail neighbor that it decides to be its child in the tree. Otherwise, the node remains active.
\end{enumerate}

Since each active node introduces its current neighbors to each other, and each node that decides to become inactive is adjacent to two tails who will remain active in the next iteration, each active neighbor in the next iteration will know its new neighbors. This process stops when there is only one node left, who will become the root of the tree. This process can be seen as a bottom-up version of pointer jumping, with randomness to deciding the next level of nodes in the tree. 

\begin{lemma}
\label{lm:iter}
    In a cycle of $n$ nodes, the above process terminates in $O(\log n )$ iterations \whp.
\end{lemma}

\begin{proof}
    Suppose before the $i$-th iteration there are currently $n_i$ active nodes. 
    For each active node $v$, let $X_v$ be the indicator random variable for the event that $v$ becomes inactive after the $i$-th iteration. Since $\Pr[X_v] = \frac{1}{8}$ and $\EE[\sum_{v \text{ is active}} X_v] =  \EE[X_v] n_i = \frac{n_i}{8}$, after each iteration there are in expectation $\frac{7}{8}$ fractions of current active nodes remain active. Hence, the process terminates in $O(\log n)$ iterations \whp by similar arguments to \Cref{lm:termination}.
\end{proof}

\paragraph{Correctness:} With \Cref{lm:iter}, and with the observation that each active node gains at most two children in each iteration, we know that the tree in the end has $O(\log n)$ degree and $O(\log n)$ depth \whp. This would be enough in our application for broadcast, since each node in $\HYBRID(\log n, 1, \log n)$ can communicate with $O(\log n)$ global neighbors in one round. 

\yanyu{this is enough but see if want to analyze for constant degree}

\paragraph{Total Message:} Suppose there are $n_i$ active nodes in round $i$. Then there are at most $4 n_i$ messages being sent. Then conditioned on the process terminating at iteration $t=c\log n$ for some constant $c$, the total number of messages sent is $\sum_{i=1}^{t} 4 n_i = O(n)$. The last step follows if each $n_{i+1}$ is at most a constant fraction of $n_i$ \whp. We know that $\EE[n_{i+1}] = \frac{7 n_i}{8}$. Since changing each coin flip can at most affect the active status of itself and its two neighbors, we can apply the McDiarmid inequality with the bounded difference parameter as $3$. We have
\begin{align*}
    \Pr\left[n_{i+1} \geq \frac{64}{63} \EE[n_{i+1}] \mid n_i  \right] &\leq \exp\left(-\frac{2 (\frac{64}{63} \EE[n_{i+1}])^2 }{ 9n_i}  \right)\\
    \Pr\left[n_{i+1} \geq \frac{8}{9} n_i \mid n_i  \right] &\leq \exp\left(-\frac{2 (\frac{8}{9} n_i)^2 }{ 9n_i}  \right)\\
     &\leq \exp\left(-\frac{n_i }{ 10}  \right).
\end{align*}
Hence, when $n_j \geq 10\log n$, the number of active nodes reduces by a factor of at least  $\frac{1}{9}$ \whp. When $n_j <10\log n$, since the number of active iterations is $O(\log n)$ \whp, we get at most $O(\log^2 n)$ messages after the $j$-th round. Hence, overall we have $O(n)$ messages \whp.


\paragraph{Node-wise message complexity for $O(\log n)$ repetitions of \RCtT{}:} \label{pa:nodewiseRCtT}
Now we analyze from each node's perspective: It can participate in at most $O(\log n)$ runs of \RCtT{} each having $O(\log n)$ iterations of coin flips \whp. A direct analysis will give us total $O(\log^2 n)$ messages since each coin flip iteration corresponds to $O(1)$ messages from each node. Because in each iteration, an active node becomes inactive with probability $\frac{1}{8}$, the number of active iterations of a node in each run of \RCtT{} follows a geometric random variable with probability $\frac{1}{8}$ and the total number of active iterations in $O(\log n)$ runs of \RCtT{} follows a negative binomial random variable with parameters $p=\frac{1}{8}$ and $r = c\log n$ for some constant $c$. Hence, we have that any fixed node spends at most $16c\log n$ active iterations in $O(\log n)$ repetitions of \RCtT{} \whp.


\end{document}